\documentclass[final]{siamltex}

\usepackage{makeidx}
\usepackage{amsfonts}
\usepackage{cite}
\usepackage{url}
\usepackage{epsfig,dsfont,amsmath,amssymb}
\usepackage{wrapfig,floatflt}
\usepackage{graphics}
\usepackage{subfig}
\usepackage{stfloats}
\usepackage{multirow}
\usepackage{footnote}
\makesavenoteenv{tabular}
\graphicspath{{./images/}}

\def \bs {\boldsymbol}
\def \Pr {\mathbb{P}}

\def \I {\mathrm{I}}
\def \R {\mathbb{R}}

\def \e {\mathbf{e}}

\def \X {\bs{X}}

\def \y {\bs y}
\def \x {\bs x}

\def \hx {\hat{\x}}
\def \z {\bs z}
\def \w {\bs w}
\def \h {\bs h}

\def \S {\mathcal{S}}

\def \u {\bs u}

\def \zero {\bs 0}

\def \E {\mathbb E}

\def \df {\stackrel{\mathrm{def}}{=}}

\def \supp {\mathrm{supp}}
\def \S {\mathcal{S}}
\def \bone {\bs 1}

\def \sup {\mathrm{sup}}

\def \nn {\nonumber}
\def \ker {\mathrm{Ker}}

\def \sI {\mathcal{I}}

\newtheorem{defi}{Definition}

\title{Verifiable and computable performance analysis of  sparsity recovery\thanks{This work was supported by ONR Grant N000140810849, and NSF Grants CCF-1014908 and CCF-0963742. A preliminary version of this work appeared in \emph{Verifiable and computable $\ \ell_\infty$ performance evaluation ofÊ$\ \ell_1$ sparse signal recovery}, in Proceedings of the 45th Annual Conference on Information Sciences and Systems, 2011.}}

\author{Gongguo Tang\thanks{Preston M. Green Department of Electrical and Systems Engineering, Washington University in St. Louis, St. Louis, MO 63130-1127, ({\tt gt2@ese.wustl.edu}).}
        \and Arye Nehorai\thanks{Preston M. Green Department of Electrical and Systems Engineering, Washington University in St. Louis, St. Louis, MO 63130-1127, ({\tt nehorai@ese.wustl.edu}).}}

\begin{document}

\maketitle

\begin{abstract}
In this paper, we develop verifiable and computable performance analysis of sparsity recovery. We define a family of goodness measures for arbitrary sensing matrices as a set of optimization problems, and design algorithms with a theoretical global convergence guarantee to compute these goodness measures. The proposed algorithms solve a series of second-order cone programs, or linear programs. As a by-product, we implement an efficient algorithm to verify a sufficient condition for exact sparsity recovery in the noise-free case. We derive performance bounds on the recovery errors in terms of these goodness measures. We also analytically demonstrate that the developed goodness measures are non-degenerate for a large class of random sensing matrices, as long as the number of measurements is relatively large. Numerical experiments show that, compared with the restricted isometry based performance bounds, our error bounds apply to a wider range of problems and are tighter, when the sparsity levels of the signals are relatively low.
\end{abstract}

\begin{keywords} 
 compressive sensing, computable performance analysis, fixed point theory, linear programming, second-order cone programming, sparsity recovery
\end{keywords}

\begin{AMS}
47H10, 90C05, 90C25, 90C26, 90C90, 94A12
\end{AMS}

\pagestyle{myheadings}
\thispagestyle{plain}
\markboth{GONGGUO TANG AND ARYE NEHORAI}{PERFORMANCE ANALYSIS OF SPARSITY RECOVERY}

\section{Introduction}
\label{sec:intro}
\noindent

Sparse signal recovery (or compressive sensing) has revolutionized the way we think of signal sampling \cite{Candes2008IntroCS}. It goes far beyond sampling and has also been applied to areas as diverse as medical imaging, remote sensing, radar, sensor arrays, image processing, computer vision, and so on. Mathematically, sparse signal recovery aims to reconstruct a sparse signal, namely a signal with only a few non-zero components, from usually noisy linear measurements:
\begin{eqnarray}
  \y &=& A\x + \w,
\end{eqnarray}
where $\x \in \R^n$ is the sparse signal, $\y \in \R^m$ is the measurement vector, $A \in \R^{m\times n}$ is the sensing/measurement matrix, and $\w \in \R^m$ is the noise. A theoretically justified way to exploit the sparseness in recovering $\x$ is to minimize its $\ell_1$ norm under certain constraints \cite{Candes2006Uncertainty}.

In this paper, we investigate the problem of using the $\ell_\infty$ norm as a performance criterion for sparse signal recovery via $\ell_1$ minimization. Although the $\ell_2$ norm has been used as the performance criterion by the majority of published research in sparse signal recovery, the adoption of the $\ell_\infty$ norm is well justified. Other popular performance criteria, such as the $\ell_1$ and $\ell_2$ norms of the error vectors, can all be expressed in terms of the $\ell_\infty$ norm in a tight and non-trivial manner. More importantly, the $\ell_\infty$ norm of the error vector has a direct connection with the support recovery problem. To see this, assuming we know \emph{a priori} the minimal non-zero absolute value of the components of the sparse signal, then controlling the $\ell_\infty$ norm within half of that value would guarantee exact recovery of the support. Support recovery is arguably one of the most important and challenging problems in sparse signal recovery. In practical applications, the support is usually physically more significant than the component values. For example, in radar imaging using sparse signal recovery, the sparsity constraints are usually imposed on the discretized time--frequency domain. The distance and velocity of a target have a direct correspondence to the support of the sparse signal. The magnitude determined by coefficients of reflection is of less physical significance\cite{Baraniuk2007Radar,Herman2008Radar,Herman2008Highradar}. Refer to \cite{Tang2010Inf} for more discussions on sparse support recovery.

Another, perhaps more important, reason to use the $\ell_\infty$ norm as a performance criterion is the verifiability and computability of the resulting performance bounds. A general strategy to study the performance of sparse signal recovery is to define a measure of the goodness of the sensing matrix, and then derive performance bounds in terms of the goodness measure. The most well-known goodness measure is undoubtedly the restricted isometry constant (RIC)\cite{Candes2008RIP}. Upper bounds on the $\ell_2$ and $\ell_1$ norms of the error vectors for various recovery algorithms have been expressed in terms of the RIC. Unfortunately, it is extremely difficult to verify that the RIC of a specific sensing matrix satisfies the conditions for the bounds to be valid, and even more difficult to directly compute the RIC itself. Actually, the only known sensing matrices with nice RICs are certain types of random matrices \cite{Juditsky2010Verifiable}. By using the $\ell_\infty$ norm as a performance criterion, we develop a framework in which a family of goodness measures for the sensing matrices are verifiable and computable. The computability further justifies the connection of the $\ell_\infty$ norm with the support recovery problem, since for the connection described in the previous paragraph to be practically useful, we must be able to compute the error bounds on the $\ell_\infty$ norm.

The verifiability and computability open doors for wide applications. In many practical applications of sparse signal recovery, \emph{e.g.}, radar imaging \cite{Sen2011Multi}, sensor arrays \cite{Willsky2005Source}, DNA microarrays \cite{Baraniuk2007DNA}, and MRI \cite{lustig2007sparse}, it is beneficial to know the performance of the sensing system before its implementation and the taking of measurements. In addition, in these application areas, we usually have the freedom to optimally design the sensing matrix. For example, in MRI the sensing matrix is determined by the sampling trajectory in the Fourier domain; in radar systems the optimal sensing matrix design is connected with optimal waveform design, a central topic of radar research. To optimally design the sensing matrix, we need to
\begin{enumerate}
  \item analyze how the performance of recovering $\x$ from $\y$ is affected by $A$, and define a function $\omega(A)$ to accurately quantify the goodness of $A$ in the context of sparse signal reconstruction;
  \item develop algorithms to efficiently verify that $\omega(A)$ satisfies the conditions for the bounds to hold, as well as to efficiently compute $\omega(A)$ for arbitrarily given $A$;
  \item design mechanisms to select within a matrix class the sensing matrix that is optimal in the sense of best $\omega(A)$.
\end{enumerate}
In this paper, we successfully address the first two points in the $\ell_\infty$ performance analysis framework.

We now preview our contributions. First of all, we propose using the $\ell_\infty$ norm as a performance criterion for sparse signal recovery and establish its connections with other performance criteria. We define a family of goodness measures of the sensing matrix, and use them to derive performance bounds on the $\ell_\infty$ norm of the recovery error vector. Performance bounds using other norms are expressed using the $\ell_\infty$ norm. Numerical simulations show that these bounds are tighter than the RIC based bounds when the sparsity levels of the signals are relatively small. Secondly and most importantly, using fixed point theory, we develop algorithms to efficiently compute the goodness measures for given sensing matrices by solving a series of second-order cone programs or linear programs, depending on the specific goodness measure being computed. We analytically demonstrate the algorithms' convergence to the global optima from any initial point. As a by-product, we obtain a fast algorithm to verify the sufficient condition guaranteeing exact sparse recovery via $\ell_1$ minimization. Finally, we show that the goodness measures are non-degenerate for subgaussian and isotropic random sensing matrices as long as the number of measurements is relatively large, a result parallel to that of the RIC for random matrices.

Several attempts have been made to address the verifiability and computability of performance analysis for sparse signal recovery, mainly based on the RIC \cite{Candes2006Uncertainty, Candes2008RIP} and the Null Space Property (NSP) \cite{Cohen2009NSP}. Due to the difficulty of explicitly computing the RIC and verifying the NSP, researchers use relaxation techniques to approximate these quantities. Examples include semi-definite programming relaxation \cite{dAspremont2007sparsePCA, dAspermont2010Nullspace} and linear programming relaxation \cite{Juditsky2010Verifiable}. To the best of the authors' knowledge, the algorithms of \cite{dAspermont2010Nullspace} and \cite{Juditsky2010Verifiable} represent state-of-the-art techniques in verifying the sufficient condition of unique $\ell_1$ recovery. In this paper, we directly address the computability of the performance bounds. More explicitly, we define the goodness measures of the sensing matrices as optimization problems and design efficient algorithms with theoretical convergence guarantees to solve the optimization problems. An algorithm to verify a sufficient condition for exact $\ell_1$ recovery is obtained only as a by-product. Our implementation of the algorithm performs orders of magnitude faster than the state-of-the-art techniques in \cite{dAspermont2010Nullspace} and \cite{Juditsky2010Verifiable},  consumes much less memory, and produces comparable results.

The paper is organized as follows. In Section \ref{sec:model}, we introduce notations, and we present the measurement model, three convex relaxation algorithms, and the sufficient and necessary condition for exact $\ell_1$ recovery. In section \ref{sec:bounds}, we derive performance bounds on the $\ell_\infty$ norms of the recovery errors for several convex relaxation algorithms. In Section \ref{sec:computation}, we design algorithms to verify a sufficient condition for exact $\ell_1$ recovery in the noise-free case, and to compute the goodness measures of arbitrarily given sensing matrices. Section \ref{sec:random} is devoted to the probabilistic analysis of our $\ell_\infty$ performance measures.  We evaluate the algorithms' performance in Section \ref{sec:numerical}. Section \ref{sec:conclusions} summarizes our conclusions.

\section{Notations, Measurement Model, and Recovery Algorithms}\label{sec:model}

In this section, we introduce notations and the measurement model, and review recovery algorithms based on $\ell_1$ minimization.

For any vector $\x \in \R^n$, the norm $\|\x\|_{k,1}$ is the summation of the absolute values of the $k$ (absolutely) largest components of $\x$. In particular, the $\ell_\infty$ norm $\|\x\|_\infty = \|\x\|_{1,1}$ and the $\ell_1$ norm $\|\x\|_1 = \|\x\|_{n,1}$. The classical inner product in $\R^n$ is denoted by $\left<\cdot, \cdot\right>$, and the $\ell_2$ (or Euclidean) norm is $\|\x\|_2 = \sqrt{\left<\x,\x\right>}$. We use $\|\cdot\|_\diamond$ to denote a general norm.

The support of $\x$, $\supp(\x)$, is the index set of the non-zero components of $\x$. The size of the support, usually denoted by the $\ell_0$ ``norm" $\|\x\|_0$,  is the sparsity level of $\x$. Signals of sparsity level at most $k$ are called $k-$sparse signals. If $S \subset \{1,\cdots,n\}$ is an index set, then $|S|$ is the cardinality of $S$, and $\x_S \in \R^{|S|}$ is the vector formed by the components of $\x$ with indices in $S$.

We use $\e_i$, $\zero$, $\bs O$, and $\bone$ to denote respectively the $i$th canonical basis vector, the zero column vector, the zero matrix, and the column vector with all ones.

Suppose $\x$ is a $k-$sparse signal. In this paper, we observe $\x$ through the following linear model:
\begin{eqnarray}\label{eqn:model}
  \y &=& A\x + \w,
\end{eqnarray}
where $A \in \R^{m\times n}$ is the measurement/sensing matrix, $\y$ is the measurement vector, and $\w$ is noise.

Many algorithms have been proposed to recover $\x$ from $\y$ by exploiting the sparseness of $\x$. We focus on three algorithms based on $\ell_1$ minimization: the Basis Pursuit \cite{Donoho1998Atomic}, the Dantzig selector \cite{Candes2007Dantzig}, and the LASSO estimator \cite{Tibshirani1996Lasso}.
\begin{eqnarray}
\hskip -1cm &&\text{Basis Pursuit:}\min_{\z \in \R^n}\|\z\|_1 \text{\ \ s.t.\ } \|\y - A\z\|_\diamond \leq \varepsilon\label{bp}\\
\hskip -1cm && \text{Dantzig:} \min_{\z \in \R^n}\|\z\|_1 \text{\ \ s.t. \ } \|A^T(\y - A\z)\|_\infty \leq \mu\label{ds}\\
\hskip -1cm &&\text{LASSO:} \min_{\z \in \R^n} \frac{1}{2}\|\y - A\z\|_2^2 + \mu \|\z\|_1 \label{lasso}.
\end{eqnarray}
Here $\mu$ is a tuning parameter, and $\varepsilon$ is a measure of the noise level. All three optimization problems have efficient implementations using convex programming or even linear programming.

In the noise-free case where $\w = 0$, roughly speaking all the three algorithms reduce to
\begin{eqnarray}\label{eqn:l1relaxation}
\min_{\z \in \R^n}\|\z\|_1 \text{\ \ s.t.\ } A\z = A\x,
\end{eqnarray}
which is the $\ell_1$ relaxation of the NP hard $\ell_0$ minimization problem:
\begin{eqnarray}
  \min_{\z \in \R^n}\|\z\|_0 \text{\ \ s.t.\ } A\z = A\x.
\end{eqnarray}

A minimal requirement on $\ell_1$ minimization algorithms is the \emph{uniqueness and exactness} of the solution $\hx \df \mathrm{argmin}_{\z: A\z = A\x} \|\x\|_1$, \emph{i.e.}, $\hx = \x$. When the true signal $\x$ is $k-$sparse, the sufficient and necessary condition for exact $\ell_1$ recovery is \cite{zhang2005overunder, donoho2001uncertainty, donoho2004highdimensional}
\begin{eqnarray}\label{nullspaceproperty}
  \sum_{i\in S}|\z_i| < \sum_{i \notin S} |\z_i|, \forall \z \in \ker(A), |S| \leq k,
\end{eqnarray}
where $\ker(A) \df \{\z: A\z = 0\}$ is the kernel of $A$, and $S \subset \{1,\ldots,n\}$ is an index set. Expressed in terms of $\|\cdot\|_{k,1}$, the necessary and sufficient condition becomes
\begin{eqnarray}\label{eqn:suff}
  \|\z\|_{k,1} < \frac{1}{2} \|\z\|_1, \forall \z \in \ker(A).
\end{eqnarray}

The approaches in \cite{Juditsky2010Verifiable} and \cite{dAspermont2010Nullspace} for verifying the sufficient condition \eqref{eqn:suff} are based on relaxing the following optimization problem in various ways:
\begin{eqnarray}\label{eqn:alphak}
  \alpha_k &=& \max_{\z} \|\z\|_{k,1} \text{\ s.t.\ } A\z = 0, \|\z\|_1 \leq 1.
\end{eqnarray}
Clearly, $\alpha_k < 1/2$ is necessary and sufficient for exact $\ell_1$ recovery for $k-$sparse signals. Unfortunately, the direct computation of \eqref{eqn:alphak} for general $k$ is extremely difficult: it is the maximization of a norm (convex function) over a polyhedron (convex set) \cite{Bodlaender1990Normmaximization}. In \cite{Juditsky2010Verifiable}, in a very rough sense $\alpha_1$ was computed by solving $n$ linear programs:
\begin{eqnarray}\label{eqn:dual}
  \min_{\y_i \in \R^m} \|\e_i - A^T\y_i\|_\infty, i = 1,\cdots,n,
\end{eqnarray}
where $\e_i$ is the $i$th canonical basis in $\R^n$. This, together with the observation that $\alpha_k < k\alpha_1$, yields an efficient algorithm to verify \eqref{eqn:suff}. However, in \cite{tang2011cmsv}, we found that the primal-dual method of directly solving \eqref{eqn:suff} as the following $n$ linear programs
\begin{eqnarray}
  \max \z_i \text{\ s.t. \ } A\z = 0, \|\z\|_1 \leq 1
\end{eqnarray}
gives rise to an algorithm orders of magnitude faster. In the next section, we will see how the computation of $\alpha_1$ arises naturally in the context of $\ell_\infty$ performance evaluation.

\section{Performance Bounds on the $\ell_\infty$ Norms of the Recovery Errors}\label{sec:bounds}
In this section, we derive performance bounds on the $\ell_\infty$ norms of the error vectors. We first establish a theorem characterizing the error vectors for the $\ell_1$ recovery algorithms, whose proof is given in Appendix \ref{app:pf:errorcharacteristics}

\begin{proposition}\label{pro:errorcharacteristics}
Suppose $\x$ in \eqref{eqn:model} is $k-$sparse and the noise $\w$ satisfies $\|\w\|_\diamond \leq \varepsilon$, $\|A^T\w\|_\infty \leq \mu$, and $\|A^T\w\|_\infty \leq \kappa \mu, \kappa \in (0,1)$, for the Basis Pursuit, the Dantzig selector, and the LASSO estimator, respectively. Define $\h = \hx - \x$ as the error vector for any of the three $\ell_1$ recovery algorithms \eqref{bp}, \eqref{ds}, and \eqref{lasso}. Then we have
\begin{eqnarray}
c\|\h\|_{k,1} \geq  \|\h\|_1,
\end{eqnarray}
where $ c = 2$ for the Basis Pursuit and the Dantzig selector, and $c = 2/(1-\kappa)$ for the LASSO estimator.
\end{proposition}

An immediate corollary of Proposition \ref{pro:errorcharacteristics} is to bound the $\ell_1$ and $\ell_2$ norms of the error vector using the $\ell_\infty$ norm:
\begin{corollary}\label{cor:connections}
Under the assumptions of Proposition \ref{pro:errorcharacteristics}, we have
\begin{eqnarray}
  \|\h\|_1 & \leq & ck \|\h\|_\infty,\label{eqn:l1linf}\\
  \|\h\|_2 & \leq & \sqrt{ck} \|\h\|_\infty. 
\label{eqn:l1linf2}
\end{eqnarray}
Furthermore, if $S = \supp(\x)$ and $\beta = \min_{i\in S}|\x_i|$, then $\|\h\|_\infty < \beta/2$ implies
\begin{eqnarray}
  \supp(\max(|\hx|-\beta/2,0) &=& \supp(\x),
\end{eqnarray}
\emph{i.e.}, a thresholding operator recovers the signal support.
\end{corollary}

For ease of presentation, we have the following definition:
\begin{definition}\label{def:linfcmsv}
For any real number $s \in [1,n]$ and matrix $A\in \R^{m\times n}$, define
\begin{eqnarray}
  \omega_{\diamond}(Q,s) &=& \min_{\z: \|\z\|_1/\|\z\|_\infty \leq s} \frac{\|Q\z\|_\diamond}{\|\z\|_\infty},
\end{eqnarray}
where $Q$ is either $A$ or $A^TA$.
\end{definition}

Now we present the error bounds on the $\ell_\infty$ norm of the error vectors for the Basis Pursuit, the Dantzig selector, and the LASSO estimator.
\begin{theorem}\label{thm:errorbound}
Under the assumption of Proposition \ref{pro:errorcharacteristics}, we have
\begin{eqnarray}
  \|\hx - \x\|_\infty \leq \frac{2\varepsilon}{\omega_\diamond(A,2k)}
\end{eqnarray}
for the Basis Pursuit,
\begin{eqnarray}
  \|\hx-\x\|_\infty \leq \frac{2\mu}{\omega_\infty(A^TA,2k)}
\end{eqnarray}
for the Dantzig selector, and
\begin{eqnarray}
  \|\hx-\x\|_\infty \leq \frac{(1+\kappa)\mu}{\omega_\infty(A^TA,2k/(1-\kappa))}
\end{eqnarray}
for the LASSO estimator.
\end{theorem}

\begin{proof}
Observe that for the Basis Pursuit
\begin{eqnarray}
\|A(\hx - \x)\|_2 &\leq& \|\y - A\hx\|_2 + \|\y - A\x\|_2\nn\\
&\leq& \varepsilon + \|A\w\|_2\nn\\
&\leq& 2\varepsilon,
\end{eqnarray}
and similarly,
\begin{eqnarray}
  \|A^TA(\hx - \x)\|_\infty \leq 2\mu
\end{eqnarray}
for the Dantzig selector, and
\begin{eqnarray}
  \|A^TA(\hx - \x)\|_\infty &\leq& (1+\kappa)\mu
\end{eqnarray}
for the LASSO estimator.
The conclusions of Theorem \ref{thm:errorbound} follow from equations \eqref{eqn:l1linf}, \eqref{eqn:l1linf2}, and Definition \ref{def:linfcmsv}.
\end{proof}

One of the primary contributions of this work is the design of algorithms that compute $\omega_\diamond(A,s)$ and $\omega_\infty(A^TA,s)$ efficiently. The algorithms provide a way to numerically assess the performance of the Basis Pursuit, the Dantzig selector, and the LASSO estimator according to the bounds given in Theorem \ref{thm:errorbound}. According to Corollary \ref{cor:connections}, the correct recovery of signal support is also guaranteed by reducing the $\ell_\infty$ norm to some threshold. In Section \ref{sec:random}, we also demonstrate that the bounds in Theorem \ref{thm:errorbound} are non-trivial for a large class of random sensing matrices, as long as $m$ is relatively large. Numerical simulations in Section \ref{sec:numerical} show that in many cases the error bounds on the $\ell_2$ norms based on Corollary \ref{cor:connections} and Theorem \ref{thm:errorbound} are tighter than the RIC based bounds. We expect the bounds on the $\ell_\infty$ norms in Theorem \ref{thm:errorbound} are even tighter, as we do not need the relaxation in Corollary \ref{cor:connections}.

We note that a prerequisite for these bounds to be valid is the positiveness of the involved $\omega_\diamond(\cdot)$. We call the validation of $\omega_\diamond(\cdot) > 0$ the verification problem. Note that from Theorem \ref{thm:errorbound}, $\omega_\diamond(\cdot) > 0$ implies the exact recovery of the true signal $\x$ in the noise-free case. Therefore, verifying $\omega_\diamond(\cdot) > 0$ is equivalent to verifying a sufficient condition for exact $\ell_1$ recovery.

\section{Verification and Computation of $\omega_\diamond$}\label{sec:computation}
In this section, we present algorithms for verification and computation of $\omega_\diamond(\cdot)$. We will present a very general algorithm and make it specific only when necessary. For this purpose, we use $Q$ to denote either $A$ or $A^TA$, and use $\|\cdot\|_\diamond$ to denote a general norm.
\subsection{Verification of $\omega_\diamond > 0$}
Verifying $\omega_\diamond(Q,s) > 0$ amounts to making sure $\|\z\|_1/\|\z\|_\infty \leq s$ for all $\z$ such that $Q\z = 0$. Equivalently, we can compute
\begin{eqnarray}\label{eqn:s_star}
  s_* &=& \min_{\z} \frac{\|\z\|_1}{\|\z\|_\infty} \text{\ s.t. \ } Q\z = 0.
\end{eqnarray}
Then, when $s < s_*$, we have $\omega_\diamond(Q,s) > 0$.
We rewrite the optimization \eqref{eqn:s_star} as
\begin{eqnarray}\label{eqn:max_inf_Q_1}
  \frac{1}{s_*} = \max_{\z}\|\z\|_\infty \text{\ s.t. \ } Q\z = 0, \|\z\|_1 \leq 1,
\end{eqnarray}
which is solved using the following $n$ linear programs:
\begin{eqnarray}\label{eqn:nlinear}
  \max_{\z} \z_i \text{\ s.t. \ } Q\z = 0,  \|\z\|_1 \leq 1.
\end{eqnarray}
The dual problem for \eqref{eqn:nlinear} is
\begin{eqnarray}\label{eqn:nlinear_dual}
  \min_{\bs \lambda} \|\e_i - Q^T\bs \lambda\|_\infty,
\end{eqnarray}
where $\e_i$ is the $i$th canonical basis vector.

We solve \eqref{eqn:nlinear} using the primal-dual algorithm expounded in Chapter 11 of \cite{boyd2004convex}, which gives an implementation much more efficient than the one for solving its dual \eqref{eqn:nlinear_dual} in \cite{Juditsky2010Verifiable}. This method is also used to implement the $\ell_1$ MAGIC for sparse signal recovery \cite{romberg2005l1magic}. Due to the equivalence of $A^TA\z = 0$ and $A\z = 0$, we always solve \eqref{eqn:max_inf_Q_1} for $Q = A$ and avoid $Q = A^TA$. The former apparently involves solving linear programs of smaller size. In practice, we usually replace $A$ with the matrix with orthogonal rows obtained from the economy-size QR decomposition of $A^T$.

As a dual of \eqref{eqn:nlinear_dual}, \eqref{eqn:nlinear} (and hence \eqref{eqn:max_inf_Q_1} and \eqref{eqn:s_star}) shares the same limitation as \eqref{eqn:nlinear_dual}, namely, it verifies $\omega_\diamond > 0$ only for $s$ up to $2\sqrt{2m}$. We now reformulate Proposition 4 of \cite{Juditsky2010Verifiable} in our framework:
\begin{proposition}\emph{\cite[Proposition 4]{Juditsky2010Verifiable}}\label{pro:sqrtmbd}
For any $m\times n$ matrix $A$ with $n \geq 32m$, one has
\begin{eqnarray}
  s_* &=& \min \left\{\frac{\|\z\|_1}{\|\z\|_\infty}: Q\z = 0\right\} < 2\sqrt{2m}.
\end{eqnarray}
\end{proposition}
\subsection{Computation of $\omega_\diamond$}
Now we turn to one of the primary contributions of this work, the computation of $\omega_\diamond$. The optimization problem is as follows:
\begin{eqnarray}\label{eqn:max_inf_Q_diamond}
  \omega_\diamond(Q,s) = \min_{\z} \frac{\|Q\z\|_\diamond}{\|\z\|_\infty} \text{\ s.t. \ } \frac{\|\z\|_1}{\|\z\|_\infty} \leq s,
\end{eqnarray}
or equivalently,
\begin{eqnarray}\label{eqn:optimizerho}
  \frac{1}{\omega_\diamond(Q,s)} = \max_{\z} \|\z\|_\infty \text{\ s.t. \ } \|Q\z\|_\diamond \leq 1, \frac{\|\z\|_1}{\|\z\|_\infty} \leq s.
\end{eqnarray}

We will show that $1/\omega_\diamond(Q,s)$ is the unique fixed point of certain scalar function. To this end, we define functions $f_{s,i}(\eta), i = 1,\ldots, n$ and $f_s(\eta)$ over $[0, \infty)$ parameterized by $s \in (1, s_*)$:
\begin{eqnarray}\label{def:fsi}
  f_{s,i}(\eta) & \df & \max_{\z} \left\{\z_i: \|Q\z\|_\diamond \leq 1, {\|\z\|_{1}} \leq s \eta\right\}\nn\\
  & = & \max_{\z} \left\{|\z_i|: \|Q\z\|_\diamond \leq 1, {\|\z\|_{1}} \leq s \eta\right\},
\end{eqnarray}
since the domain for the maximization is symmetric to the origin, and
\begin{eqnarray}\label{def:fs}
  f_s(\eta) &\df& \max_{\z} \left\{\|\z\|_{\infty}: \|Q\z\|_\diamond \leq 1, {\|\z\|_{1}} \leq s \eta\right\}\nn\\
  & = & \max_{\substack{\z: \|Q\z\|_\diamond \leq 1\\ {\|\z\|_{1}} \leq s \eta}} \max_i |\z_i|\nn\\
  &= & \max_i \max_{\substack{\z: \|Q\z\|_\diamond \leq 1\\ {\|\z\|_{1}} \leq s \eta}} |\z_i|\nn\\
  & = & \max_i f_{s,i}(\eta),
\end{eqnarray}
where for the last but one equality we have exchanged the two maximizations. For $\eta > 0$, it is easy to show that strong duality holds for the optimization problem defining $f_{s,i}(\eta)$. As a consequence, we have the dual form of $f_{s,i}(\eta)$:
\begin{eqnarray}\label{eqn:fsi_dual}
  f_{s,i}(\eta) &=& \min_{\bs \lambda} s\eta \|\e_i - Q^T\bs \lambda\|_\infty + \|\bs \lambda\|_\diamond^*,
\end{eqnarray}
where $\|\cdot\|_\diamond^*$ is the dual norm of $\|\cdot\|_\diamond$.

In the definition of $f_s(\eta)$, we basically replaced the $\|\z\|_{\infty}$ in the denominator of the fractional constraint in \eqref{eqn:optimizerho} with $\eta$. The following theorem states that the unique positive fixed point of $f_s(\eta)$ is exactly $1/\omega_\diamond(Q,s)$. See Appendix \ref{app:pf:fix_fs} for the proof.

\begin{theorem}\label{thm:fix_fs}
The functions $f_{s,i}(\eta)$ and $f_s(\eta)$ have the following properties:
\begin{enumerate}
  \item $f_{s,i}(\eta)$ and $f_s(\eta)$ are continuous in $\eta$;
  \item $f_{s,i}(\eta)$ and $f_s(\eta)$ are strictly increasing in $\eta$;
  \item $f_{s,i}(\eta)$ is concave for every $i$;
\item $f_s(0) = 0$, $f_s(\eta) \geq s\eta > \eta$ for sufficiently small $\eta > 0$, and there exists $\rho < 1$ such that $f_s(\eta) < \rho\eta$ for sufficiently large $\eta$; the same holds for $f_{s,i}(\eta)$;
    \item $f_{s,i}$ and $f_s(\eta)$ have unique positive fixed points $\eta_i^* = f_{s,i}(\eta_i^*)$ and $\eta^* = f_s(\eta^*)$, respectively; and $\eta^* = \max_i \eta_i^*$;
    \item The unique positive fixed point of $f_{s}(\eta)$, $\eta^*$, is equal to $1/\omega_\diamond(Q,s)$;
    \item For $\eta \in (0, \eta^*)$, we have $f_s(\eta) > \eta$; and for $\eta \in (\eta^*, \infty)$, we have $f_s(\eta) < \eta$; the same statement holds also for $f_{s,i}(\eta)$.
  \item For any $\epsilon > 0$, there exists $\rho_1(\epsilon) > 1$ such that $f_s(\eta) > \rho_1(\epsilon) \eta$ as long as $0 < \eta \leq (1-\epsilon) \eta^*$; and there exists $\rho_2(\epsilon) < 1$ such that $f_s(\eta) < \rho_2(\epsilon) \eta$ as long as $\eta > (1+\epsilon) \eta^*$.
\end{enumerate}
\end{theorem}

Theorem \ref{thm:fix_fs} implies three ways to compute the fixed point of $\eta^*=1/\omega_\diamond(Q,s)$ for $f_s(\eta)$.
\begin{enumerate}
  \item \textbf{Naive Fixed Point Iteration:} Property 8) of Theorem \ref{thm:fix_fs} suggests that the fixed point iteration
      \begin{eqnarray}\label{eqn:fpiteration}
        \eta_{t+1} &=& f_s(\eta_t), t = 0, 1, \ldots
      \end{eqnarray}
      starting from any initial point $\eta_0 > 0$ converges to $\eta^*$, no matter $\eta_0 < \eta^*$ or $\eta_0 > \eta^*$. The algorithm can be made more efficient in the case $\eta_0 < \eta^*$. More specifically, since $f_s(\eta) = \max_i f_{s,i}(\eta)$, at each fixed point iteration, we set $\eta_{t+1}$ to be the first $f_{s,i}(\eta_t)$ that is greater than $\eta_t + \epsilon$ with $\epsilon$ some tolerance parameter. If for all $i$, $f_{s,i}(\eta_t) < \eta_t + \epsilon$, then $f_s(\eta_t) = \max_i f_{s,i}(\eta_t) < \eta_t + \epsilon$, which indicates the optimal function value can not be improved greatly and the algorithm should terminate. In most cases, to get $\eta_{t+1}$, we need to solve only one optimization problem $\max_{\z} \z_i: \|Q\z\|_\diamond \leq 1, {\|\z\|_{1}} \leq s \eta_t$ instead of $n$. This is in contrast to the case where $\eta_0 > \eta^*$, because in the later case we must compute all $f_{s,i}(\eta_t)$ to update $\eta_{t+1} = \max_i f_{s,i}(\eta_t)$. An update based on a single $f_{s,i}(\eta_t)$ might generate a value smaller than $\eta^*$.

\begin{figure}[h!t]
\hskip 0cm
\centering
\includegraphics[width = 0.7\textwidth, trim = 0mm 0mm 0mm 0mm, clip]{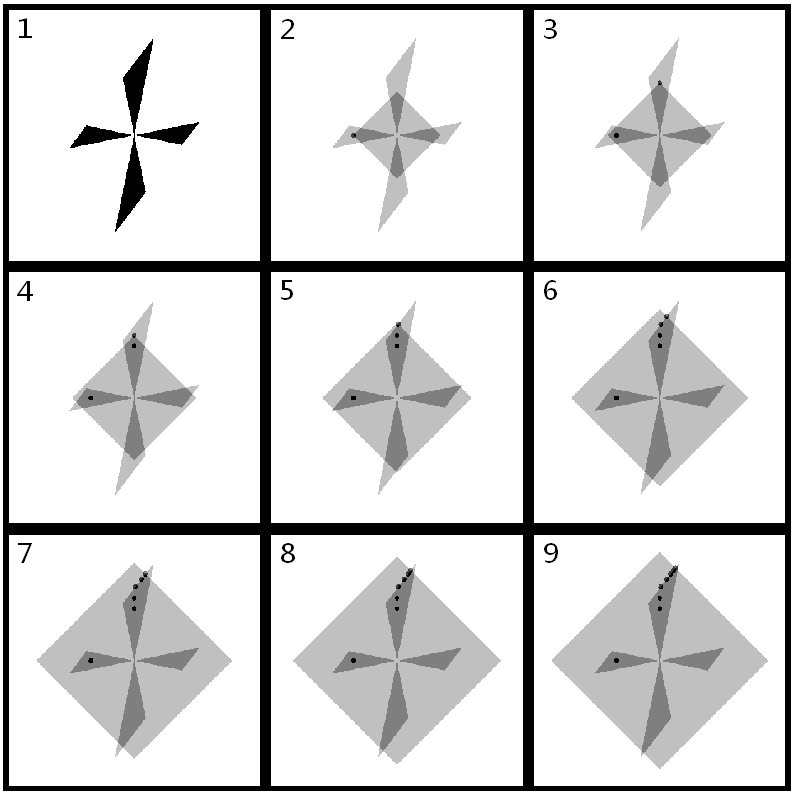}
\caption{Illustration of the naive fixed point iteration \eqref{eqn:fpiteration} when $\diamond = \infty$.}
\label{fig:convergence}
\end{figure}%
In Figure \ref{fig:convergence}, we illustrate the behavior of the naive fixed point iteration algorithm \eqref{eqn:fpiteration}. These figures are generated by Matlab for a two dimensional problem. We index the sub-figures from left to right and from top to bottom. The first (upper left) sub-figure shows the star-shaped region $\S = \{\z: \|Q\z\|_\infty \leq 1, \|\z\|_1/\|\z\|_\infty \leq s\}$. Starting from an initial $\eta_0 < \eta^*$, the algorithm solves
\begin{eqnarray}\label{eqn:fig0}
  \max_{\z} \|\z\|_\infty \text{\ s.t. \ } \|Q\z\|_\diamond \leq 1, \|\z\|_1 \leq s\eta_0
\end{eqnarray}
in sub-figure 2. The solution is denoted by the black dot. Although the true domain for the optimization in \eqref{eqn:fig0} is the intersection of the distorted $\ell_\infty$ ball $\{\z: \|Q\z\|_\infty \leq 1\}$ and the $\ell_1$ ball $\{\z: \|\z\|_1 \leq s\eta_0\}$, the intersection of the $\ell_1$ ball (light gray diamond) and the star-shaped region $\S$ forms the effective domain, which is the dark grey region in the sub-figures. To see this, we note the optimal value of the optimization \eqref{eqn:fig0} $\eta_1 = \|\x_1^*\|_\infty = f_s(\eta_0) > \eta_0$ according to 7) of Theorem \ref{thm:fix_fs}, implying that, for the optimal solution $\x_1^*$, we have $\|\x_1^*\|_1/\|\x_1^*\|_\infty \leq \|\x_1^*\|_1/\eta_0 \leq s$. Therefore, the optimal solution $\x_1^*$ can always be found in the dark grey region. In the following sub-figures, at each iteration, we expand the $\ell_1$ ball until we get to the tip point of the star-shaped region $\S$, which is the global optimum.

      Despite of its simplicity, the naive fixed point iteration has two major disadvantages. Firstly, the stopping criterion based on successive improvement is not accurate as it does not reflect the gap between $\eta_t$ and $\eta^*$. This disadvantage can be remedied by starting from both below and above $\eta^*$. The distance between corresponding terms in the two generated sequences is an indication of the gap to the fixed point $\eta^*$. However, the resulting algorithm is generally slow, especially when updating $\eta_{t+1}$ from above $\eta^*$. Secondly, the iteration process is slow when close to the fixed point $\eta^*$. This is because $\rho_1(\epsilon)$ and $\rho_2(\epsilon)$ in 8) of Theorem \ref{thm:fix_fs} are close to 1 for small $\epsilon > 0$.
  \item \textbf{Bisection:} The bisection approach is motivated by property 7) of Theorem \ref{thm:fix_fs}. Starting from an initial interval $(\eta_\mathrm{L}, \eta_\mathrm{U})$ that contains $\eta^*$, we compute $f_s(\eta_\mathrm{M})$ with $\eta_\mathrm{M} = (\eta_{\mathrm{L}} + \eta_\mathrm{U})/2$. As a consequence of property 7), $f_s(\eta_\mathrm{M}) > \eta_\mathrm{M}$ implies $f_s(\eta_\mathrm{M}) < \eta^*$, and we set $\eta_\mathrm{L} = f_s(\eta_\mathrm{M})$; $f_s(\eta_\mathrm{M}) < \eta_\mathrm{M}$ implies $f_s(\eta_\mathrm{M}) > \eta^*$, and we set $\eta_\mathrm{U} = f_s(\eta_\mathrm{M})$. The bisection process can also be accelerated by setting $\eta_\mathrm{L} = f_{s,i}(\eta_\mathrm{M})$ for the first $f_{s,i}(\eta_\mathrm{M})$ greater than $\eta_\mathrm{M}$. The convergence of the bisection approach is much faster than the naive fixed point iteration because each iteration reduces the interval length at least by half. In addition, half the length of the interval is an upper bound on the gap between $\eta_{\mathrm{M}}$ and $\eta^*$, resulting an accurate stopping criterion. However, if the initial $\eta_{\mathrm{U}}$ is too larger than $\eta^*$, the majority of $f_s(\eta_{\mathrm{M}})$ would turn out to be less than $\eta^*$. The verification of $f_s(\eta_{\mathrm{M}}) < \eta_{\mathrm{M}}$ needs solving $n$ linear programs or second-order cone programs, greatly degrading the algorithm's performance.
  \item \textbf{Fixed Point Iteration $+$ Bisection:} The third approach combines the advantages of the bisection method and the fixed point iteration method, at the level of $f_{s,i}(\eta)$. This method relies on the representation $f_s(\eta) = \max_i f_{s,i}(\eta)$ and $\eta^* = \max_i \eta_i^*$.

      Starting from an initial interval $(\eta_{\mathrm{L}0}, \eta_\mathrm{U})$ and the index set $\sI_0 = \{1,\ldots, n\}$, we pick any $i_0 \in \sI_0$ and use the (accelerated) bisection method with starting interval $(\eta_{\mathrm{L}0}, \eta_{\mathrm{U}})$ to find the positive fixed point $\eta_{i_0}^*$ of $f_{s,i_0}(\eta)$. For any $i \in \sI_0/i_0$, $f_{s,i}(\eta_{i_0}^*) \leq \eta_{i_0}^*$ implies that the fixed point $\eta_i^*$ of $f_{s,i}(\eta)$ is less than or equal to $\eta_{i_0}^*$ according to the continuity of $f_{s,i}(\eta)$ and the uniqueness of its positive fixed point. As a consequence, we remove this $i$ from the index set $\sI_0$. We denote $\sI_1$ as the index set after all such $i$s removed, \emph{i.e.,} $\sI_1 = \sI_0/\{i: f_{s,i}(\eta_{i_0}^*) \leq \eta_{i_0}^*\}$. We then set $\eta_{\mathrm{L1}} = \eta_{i_0}^*$ as $\eta^* \geq \eta_{i_0}^*$. Next we test the $i_1 \in \sI_1$ with the \emph{largest} $f_{s,i}(\eta_{i_0}^*)$ and construct $\sI_2$ and $\eta_{\mathrm{L}2}$ in a similar manner. We repeat the process until the index set $\sI_t$ is empty. The $\eta_i^*$ found at the last step is the maximal $\eta_i^*$, which is equal to $\eta^*$.
\end{enumerate}

\vspace{.3cm}
Note that in equations \eqref{eqn:max_inf_Q_diamond}, \eqref{eqn:optimizerho}, and \eqref{def:fs}, if we replace the $\ell_\infty$ norm with any other norm (with some other minor modifications), especially $\|\cdot\|_{s,1}$ or $\|\cdot\|_2$, then a naive fixed point iteration algorithm still exists. In addition, as we did in Corollary \ref{cor:connections}, we can express other norms on the error vector in terms of $\|\cdot\|_{s,1}$ and $\|\cdot\|_2$. We expect the norm $\|\cdot\|_{s,1}$ would yield the tightest performance bounds. Unfortunately, the major problem is that in these cases, the function $f_s(\eta)$ do not admit an obvious polynomial time algorithm to compute. It is very likely the corresponding norm maximization defining $f_s(\eta)$ for $\|\cdot\|_{s,1}$ and $\|\cdot\|_2$ are NP hard \cite{Bodlaender1990Normmaximization}.

\section{Probabilistic Behavior of $\omega_\diamond(Q,s)$}\label{sec:random}
In \cite{tang2011cmsv}, we defined the $\ell_1$-constrained minimal singular value ($\ell_1$-CMSV) as a goodness measure of the sensing matrix and established performance bounds using $\ell_1-$CMSV. For comparison, we include the definition below:
\begin{defi}\label{def:l1cmsv}
For any $s \in [1,n]$ and matrix $A\in \R^{m\times n}$, define the $\ell_1$-constrained minimal singular value (abbreviated as $\ell_1$-CMSV) of $A$ by
\begin{eqnarray}
\rho_s(A) = \min_{\z:\ {\|\z\|_1^2}/{\|\z\|_2^2} \leq s} \frac{\|A\z\|_2}{\|\z\|_2}.
\end{eqnarray}
\end{defi}

Despite the seeming resemblance of the definitions between $\omega_\diamond(Q,s)$, especially $\omega_2(A,s)$, and $\rho_s(A)$, the difference in the $\ell_\infty$ norm and the $\ell_2$ norm has important implications. As shown in Theorem \ref{thm:fix_fs}, the $\ell_\infty$ norm enables the design of optimization procedures with nice convergence properties to efficiently compute $\omega_\diamond(Q,s)$. On the other hand, the $\ell_1$-CMSV yields tight performance bounds at least for a large class of random sensing matrices, as we will see in Theorem \ref{thm:randomcmsv}.

However, there are some interesting connections among these quantities, as shown in the following proposition. These connections allow us the analyze the probabilistic behavior of $\omega_\diamond(Q,s)$ using the results for $\rho_s(A)$ established in \cite{tang2011cmsv}.
\begin{proposition}
\begin{eqnarray}
   \sqrt{s}\sqrt{\omega_\infty(A^TA,s)} \geq \omega_2(A,s) \geq \rho_{s^2}(A).
\end{eqnarray}
\end{proposition}
\begin{proof}
For any $\z$ such that $\|\z\|_\infty = 1$ and $\|\z\|_1 \leq s$, we have
\begin{eqnarray}
  \z A^TA\z &\leq& \sum_{i}|\z_i||(A^TA\z)_i|\nn\\
   &\leq & \|\z\|_1\|A^TA\z\|_\infty\nn\\
   &\leq & s\|A^TA\z\|_\infty.
\end{eqnarray}
Taking the minimum over $\{\z: \|\z\|_\infty = 1, \|\z\|_1 \leq s\}$ yields
\begin{eqnarray}
  \omega_2^2(A,s) &\leq& s \omega_\infty(A^TA,s).
\end{eqnarray}
Note that $\|\z\|_1/\|\z\|_\infty \leq s$ implies $\|\z\|_1 \leq s\|\z\|_\infty \leq s \|\z\|_2$, or equivalently,
\begin{eqnarray}
  \{\z: \|\z\|_1/\|\z\|_\infty \leq s \} \subseteqq  \{\z: \|\z\|_1/\|\z\|_2 \leq s \}.
\end{eqnarray}
As a consequence, we have
\begin{eqnarray}
  \omega_2(A,s) &=& \min_{\|\z\|_1/\|\z\|_\infty \leq s}\frac{\|A\z\|_2}{\|\z\|_2} \frac{\|\z\|_2}{\|\z\|_\infty}\nn\\
  &\geq &   \min_{\|\z\|_1/\|\z\|_\infty \leq s}\frac{\|A\z\|_2}{\|\z\|_2}\nn\\
  &\geq &   \min_{\|\z\|_1/\|\z\|_2 \leq s}\frac{\|A\z\|_2}{\|\z\|_2}\nn\\
  & = & \rho_{s^2}(A),
\end{eqnarray}
where the first inequality is due to $\|\z\|_2 \geq \|\z\|_\infty$, and the second inequality is because the minimization is taken over a larger set.
\end{proof}

As a consequence of the theorem we established in \cite{tang2011cmsv} and include below, we derive a condition on the number of measurements to get $\omega_\diamond(Q,s)$ bounded away from zero with high probability for sensing matrices with \emph{i.i.d.} subgaussian and isotropic rows. Note that a random vector $\X \in \R^n$ is called \emph{isotropic and subgaussian} with constant $L$ if $\E|\left<\X, \u\right>|^2 = \|\u\|_2^2$ and $\Pr(|\left<\X, \u\right>| \geq t ) \leq 2 \exp(-t^2/(L\|\u\|_2))$ hold for any $\u \in \R^n$.
\begin{theorem}\emph{\cite{tang2011cmsv}}\label{thm:randomcmsv}
Let the rows of the scaled sensing matrix $\sqrt{m}A$ be \emph{i.i.d.} subgaussian and isotropic random vectors with numerical constant $L$. Then there exist constants $c_1$ and $c_2$ such that for any $\epsilon > 0$ and $m \geq 1$ satisfying
\begin{eqnarray}
  m \geq c_1 \frac{L^2 s \log n}{\epsilon^2},
\end{eqnarray}
we have
\begin{eqnarray}
  \E |1 - \rho_s(A)| \leq \epsilon,
\end{eqnarray}
and
\begin{eqnarray}
  \Pr\{1 - \epsilon \leq \rho_s(A) \leq 1 + \epsilon\} \geq  1 - \exp(-c_2 \epsilon^2m/L^4).
\end{eqnarray}
\end{theorem}

\begin{theorem}\label{thm:randomomega}
Under the assumptions and notations of Theorem \ref{thm:randomcmsv}, there exist constants $c_1$ and $c_2$ such that for any $\epsilon > 0$ and $m \geq 1$ satisfying
\begin{eqnarray}\label{eqn:mrandombd}
  m \geq c_1 \frac{L^2s^2 \log n}{\epsilon^2},
\end{eqnarray}
we have
\begin{eqnarray}
  &&\E\ \omega_2(A,s) \geq 1 - \epsilon,\\
  &&\Pr\{\omega_2(A,s) \geq 1 - \epsilon\} \geq  1 - \exp(-c_2 \epsilon^2m),
\end{eqnarray}
and
\begin{eqnarray}
  \hskip -1cm &&\E {\ \omega_\infty(A^TA,s)} \geq \frac{(1-\epsilon)^2}{s},\label{eqn:meanomega_inf}\\
  \hskip -1cm &&\Pr\left\{\omega_\infty(A,s) \geq \frac{(1 - \epsilon)^2}{s}\right\} \geq  1 - \exp(-c_2 \epsilon^2m)\label{eqn:probomega_inf}.
\end{eqnarray}
\end{theorem}

Sensing matrices with \emph{i.i.d.} subgaussian and isotropic rows include the Gaussian ensemble, and the Bernoulli ensemble, as well as the normalized volume measure on various convex symmetric bodies, for example, the unit balls of $\ell_p^n$ for $2 \leq p \leq \infty$ \cite{mendelson2007subgaussian}. In equations \eqref{eqn:meanomega_inf} and \eqref{eqn:probomega_inf}, the extra $s$ in the lower bound of $\omega_\infty(A^TA,s)$ would contribute an $s$ factor in the bounds of Theorem \ref{thm:errorbound}. It plays the same role as the extra $\sqrt{k}$ factor in the error bounds for the Dantzig selector and the LASSO estimator in terms of the RIC and the $\ell_1-$CMSV \cite{Candes2007Dantzig, tang2011cmsv}.

The measurement bound \eqref{eqn:mrandombd} implies that the algorithms for verifying $\omega_\diamond > 0$ and for computing $\omega_\diamond$ work for $s$ at least up to the order $\sqrt{m/\log n}$. The order $\sqrt{m/\log n}$ is complementary to the $\sqrt{m}$ upper bound in Proposition \ref{pro:sqrtmbd}.

Note that Theorem \ref{thm:randomcmsv} implies that the following program:
\begin{eqnarray}
  \max_{\z} \|\z\|_2 \text{\  s.t.\ } A\z = 0, \|\z\|_1 \leq 1,
\end{eqnarray}
verifies the sufficient condition for exact $\ell_1$ recovery for $s$ up to the order $m/\log n$, at least for subgaussian and isotropic random sensing matrices. Unfortunately, this program is NP hard and hence not tractable.

\section{Numerical Experiments}\label{sec:numerical}
In this section, we provide implementation details and numerically assess the performance of the algorithms for solving \eqref{eqn:max_inf_Q_diamond} using the naive fixed point iteration. The numerical implementation and performance of \eqref{eqn:s_star} were previously reported in \cite{tang2011cmsv} and hence are omitted here.  All the numerical experiments in this section were conducted on a desktop computer with a Pentium D CPU@3.40GHz, 2GB RAM, and Windows XP operating system, and the computations were running single-core.

Recall that the optimization defining $f_{s,i}(\eta)$ is
\begin{eqnarray}\label{eqn:nconvex}
  \min \z_i \text{\ s.t.\ } \|Q\z\|_\diamond \leq 1, \|\z\|_1 \leq s\eta.
\end{eqnarray}
Depending on whether $\diamond = 1, \infty$, or $2$, \eqref{eqn:nconvex} is solved using either linear programs or second-order cone programs. For example, when $\diamond = \infty$, we have the following corresponding linear programs:
\begin{eqnarray}\label{eqn:nlinearstandardinf}
\hskip -0.8cm &&\min \left[
       \begin{array}{cc}
         \e_i^T & \zero^T \\
       \end{array}
     \right]
\left[
                                  \begin{array}{c}
                                    \z \\
                                    \u \\
                                  \end{array}
                                \right]\nn\\
\hskip -0.8cm &&\text{\  s.t.\ }
\left[
  \begin{array}{rr}
    Q & \bs O\\
    -Q & \bs O\\
    \I & -\I \\
    -\I & -\I \\
    \zero^T & \bone^T\\
  \end{array}
\right]\left[
                                  \begin{array}{c}
                                    \z \\
                                    \u \\
                                  \end{array}
                                \right] \leq  \left[
                                                     \begin{array}{c}
                                                       \bone\\
                                                       \bone\\
                                                       \zero \\
                                                       \zero \\
                                                       s\eta \\
                                                     \end{array}
                                                   \right], i = 1, \ldots, n,
\end{eqnarray}
These linear programs are implemented using the primal-dual algorithm outlined in Chapter 11 of \cite{boyd2004convex}. The algorithm finds the optimal solution together with optimal dual vectors by solving the Karush-Kuhn-Tucker condition using linearization. The major computation is spent in solving linear systems of equations with positive definite coefficient matrices. When $\diamond = 2$, we rewrite \eqref{eqn:nconvex} as the following second-order cone programs
\begin{eqnarray}\label{eqn:nsocp}
\hskip -0.8cm &&\min \left[
       \begin{array}{cc}
         \e_i^T & \zero^T \\
       \end{array}
     \right]
\left[
                                  \begin{array}{c}
                                    \z \\
                                    \u \\
                                  \end{array}
                                \right]\nn\\
\hskip -0.8cm &&\text{\  s.t.\ } \frac{1}{2}\left(\left\|\left[
  \begin{array}{cc}
    Q & \bs O \\
  \end{array}
\right]\left[
                                  \begin{array}{c}
                                    \z \\
                                    \u \\
                                  \end{array}
                                \right]\right\|_2^2 - 1\right) \leq 0\nn\\
\hskip -0.8cm &&\ \ \ \ \ \left[
  \begin{array}{rr}
     \I & -\I \\
    -\I & -\I \\
    \zero^T & \bone^T\\
  \end{array}
\right]\left[
                                  \begin{array}{c}
                                    \z \\
                                    \u \\
                                  \end{array}
                                \right] \leq  \left[
                                                     \begin{array}{c}
                                                       \zero \\
                                                       \zero \\
                                                       s\eta \\
                                                     \end{array}
                                                   \right].
\end{eqnarray}
We use the log-barrier algorithm described in Chapter 11 of \cite{boyd2004convex} to solve \eqref{eqn:nsocp}. Interested readers are encouraged to refer to \cite{romberg2005l1magic} for a concise exposition of the general primal-dual and log-barrier algorithms and implementation details for similar linear programs and second-order cone programs.

We test the algorithms on Bernoulli, Gaussian, and Hadamard matrices of different sizes. The entries of Bernoulli and Gaussian matrices are randomly generated from the classical Bernoulli distribution with equal probability and the standard Gaussian distribution, respectively.  For Hadamard matrices, first a square Hadamard matrix of size $n$ ($n$ is a power of 2) is generated, then its rows are randomly permuted and its first $m$ rows are taken as an $m\times n$ sensing matrix. All $m\times n$ matrices are normalized to have columns of unit length.

We compare our recovery error bounds based on $\omega_\diamond$ with those based on the RIC. Combining Corollary \ref{cor:connections} and Theorem \ref{thm:errorbound}, we have for the Basis Pursuit
\begin{eqnarray}\label{eqn:bp_omega_bd}
  \|\hx - \x\|_2 \leq \frac{2\sqrt{2k}}{\omega_2(A,2k)}\varepsilon,
\end{eqnarray}
and for the Dantzig selector
\begin{eqnarray}\label{eqn:ds_omega_bd}
  \|\hx - \x\|_2 \leq \frac{2\sqrt{2k}}{\omega_\infty(A^TA,2k)}\mu.
\end{eqnarray}

For comparison, the two RIC bounds are
\begin{eqnarray}\label{eqn:bp_rip_bd}
\|\hx-\x\|_2 \leq \frac{4\sqrt{1+\delta_{2k}(A)}}{1-(1+\sqrt{2})\delta_{2k}(A)}\varepsilon,
\end{eqnarray}
for the Basis Pursuit, assuming $\delta_{2k}(A) < \sqrt{2}-1$ \cite{Candes2008RIP}, and
\begin{eqnarray}\label{eqn:ds_rip_bd}
        \|\hx-\x\|_2 \leq \frac{4\sqrt{k}}{1-\delta_{2k}(A)-\delta_{3k}(A)}\mu,
      \end{eqnarray}
for the Dantzig selector, assuming $\delta_{2k}(A) + \delta_{3k}(A) < 1$ \cite{Candes2007Dantzig}. Without loss of generality, we set $\varepsilon = 1$ and $\mu = 1$.

The RIC is computed using Monte Carlo simulations. More explicitly, for $\delta_{2k}(A)$, we randomly take $1000$ sub-matrices of $A \in \R^{m\times n}$ of size $m\times 2k$, compute the maximal and minimal singular values $\sigma_1$ and $\sigma_{2k}$, and approximate $\delta_{2k}(A)$ using the maximum of $\max(\sigma_1^2 - 1, 1 - \sigma_{2k}^2)$ among all sampled sub-matrices. Obviously, the approximated RIC is always smaller than or equal to the exact RIC. As a consequence, the performance bounds based on the exact RIC are \emph{worse} than those based on the approximated RIC. Therefore, in cases where our $\omega_\diamond$ based bounds are better (tighter, smaller) than the approximated RIC bounds, they are even better than the exact RIC bounds.

In Tables \ref{tbl:BernoulliBPRicOmega}, \ref{tbl:HadamardBPRicOmega}, and \ref{tbl:GaussianBPRicOmega}, we compare the error bounds \eqref{eqn:bp_omega_bd} and \eqref{eqn:bp_rip_bd} for the Basis Pursuit algorithm. In the tables, we also include $s_*$ computed by \eqref{eqn:max_inf_Q_1}, and $k_* = \lfloor s_*/2 \rfloor$, \emph{i.e.}, the maximal sparsity level such that the sufficient and necessary condition \eqref{nullspaceproperty} holds. The number of measurements is taken as $m = \lfloor\rho n\rfloor, \rho = 0.2, 0.3, \ldots, 0.8$. 
Note the blanks mean that the corresponding bounds are not valid. For the Bernoulli and Gaussian matrices, the RIC bounds work only for $k \leq 2$, even with $m = \lfloor 0.8n\rfloor$, while the $\omega_2(A,2k)$ bounds work up until $k = 9$. Both bounds are better for Hadamard matrices. For example, when $m = 0.5n$, the RIC bounds are valid for $k \leq 3$, and our bounds hold for $k \leq 5$. In all cases for $n = 256$, our bounds are smaller than the RIC bounds.

\begin{table}
\caption{Comparison of the $\omega_2$ based bounds and the RIC based bounds on the $\ell_2$ norms of the errors of the Basis Pursuit algorithm for a Bernoulli matrix with leading dimension $n = 256$.}
\begin{center}
\hskip -.3 cm

\begin{tabular}{||l|l||l|l|l|l|l|l|l|}
\cline{2-9}
\multicolumn{1}{l|}{\multirow{2}{*}{}} & $m$ & 51 & 77 & 102 & 128 & 154 & 179 & 205\\
\cline{2-9}
\multicolumn{1}{l|}{}& $s_*$ & 4.6 & 6.1 & 7.4& 9.6 & 12.1& 15.2& 19.3\\
\hline
$k$ & $k_*$ & 2 & 3 & 3 & 4 & 6 & 7 & 9\\
\hline\hline

\multirow{2}{*}{1} & $\omega$ bd & 4.2 & 3.8 & 3.5 & 3.4 & 3.3 & 3.2 & 3.2\\
& ric bd & & & 23.7 & 16.1 & 13.2 & 10.6 & 11.9\\
\hline

\multirow{2}{*}{2} & $\omega$ bd & 31.4 & 12.2 & 9.0 & 7.4 & 6.5 & 6.0 & 5.6\\
& ric bd & & & & & & 72.1 & 192.2\\
\hline

\multirow{2}{*}{3} & $\omega$ bd &  & 252.0 & 30.9 & 16.8 & 12.0 & 10.1 & 8.9\\
& ric bd & & & & & & & \\
\hline

\multirow{2}{*}{4} & $\omega$ bd & \multicolumn{3}{l|}{} & 52.3 &   23.4 &   16.5  &  13.6\\
& ric bd & \multicolumn{3}{l|}{} & & & & \\
\hline

\multirow{2}{*}{5} & $\omega$ bd & \multicolumn{4}{l|}{} & 57.0  &  28.6 &   20.1\\
& ric bd & \multicolumn{4}{l|}{} & & & \\
\hline

\multirow{2}{*}{6} & $\omega$ bd & \multicolumn{4}{l|}{} & 1256.6 &   53.6 &   30.8\\
& ric bd & \multicolumn{4}{l|}{} & & & \\
\hline

\multirow{2}{*}{7} & $\omega$ bd & \multicolumn{5}{l|}{} & 161.6  &  50.6\\
& ric bd & \multicolumn{5}{l|}{} & & \\
\hline

\multirow{2}{*}{8} & $\omega$ bd & \multicolumn{6}{l|}{} & 93.1\\
& ric bd & \multicolumn{6}{l|}{} & \\
\hline

\multirow{2}{*}{9} & $\omega$ bd & \multicolumn{6}{l|}{} & 258.7\\
& ric bd & \multicolumn{6}{l|}{} & \\
\hline
\end{tabular}\label{tbl:BernoulliBPRicOmega}
\end{center}
\end{table}

\begin{table}
\caption{Comparison of the $\omega_2$ based bounds and the RIC based bounds on the $\ell_2$ norms of the errors of the Basis Pursuit algorithm for a Hadamard matrix with leading dimension $n = 256$.}
\begin{center}
\hskip -.4 cm
\begin{tabular}{||l|l||l|l|l|l|l|l|l|}
\cline{2-9}
\multicolumn{1}{l|}{\multirow{2}{*}{}} & $m$ & 51 & 77 & 102 & 128 & 154 & 179 & 205\\
\cline{2-9}
\multicolumn{1}{l|}{}& $s_*$ & 5.4 & 7.1 & 9.1 & 11.4 & 14.0 & 18.4 & 25.3 \\
\hline
$k$ & $k_*$ & 2 & 3 & 4 & 5 & 6 & 9 & 12\\
\hline\hline

\multirow{2}{*}{1} & $\omega$ bd & 3.8 & 3.5 & 3.3 & 3.2 & 3.1 & 3.0 & 3.0\\
& ric bd & 46.6 &   13.2  &   9.2   &  9.4 &    8.3   &  6.2 &    5.2\\
\hline

\multirow{2}{*}{2} & $\omega$ bd & 13.7  &   8.4  &   6.7 &    5.9  &   5.4  &   4.9  &   4.6\\
& ric bd & & & 46.6  &  24.2  &  15.3  &   8.6   &  7.1\\
\hline

\multirow{2}{*}{3} & $\omega$ bd &  & 30.9 &   14.0 &   10.1  &   8.4  &   7.1 &    6.3\\
& ric bd & & & & 1356.6   & 25.4 &   10.3  &   8.8\\
\hline

\multirow{2}{*}{4} & $\omega$ bd & \multicolumn{2}{l|}{} & 47.4 &   18.9 &   13.2  &   9.9   &  8.1\\
& ric bd & \multicolumn{2}{l|}{} & & & 40.0 &   14.0 &   10.2\\
\hline

\multirow{2}{*}{5} & $\omega$ bd & \multicolumn{3}{l|}{} & 51.5   & 22.6 &   13.8 &   10.3\\
& ric bd & \multicolumn{3}{l|}{} & & & 18.8  &  11.6\\
\hline

\multirow{2}{*}{6} & $\omega$ bd & \multicolumn{4}{l|}{} & 50.8 &   20.1  &  13.1\\
& ric bd & \multicolumn{4}{l|}{} & & 42.5  &  15.9\\
\hline

\multirow{2}{*}{7} & $\omega$ bd & \multicolumn{5}{l|}{} & 31.8  &  16.7\\
& ric bd & \multicolumn{5}{l|}{}& 94.2 &   19.7\\
\hline

\multirow{2}{*}{8} & $\omega$ bd & \multicolumn{5}{l|}{} & 63.5  &  21.7\\
& ric bd & \multicolumn{5}{l|}{}& 1000.0 &   24.6\\
\hline

\multirow{2}{*}{9} & $\omega$ bd & \multicolumn{5}{l|}{} & 449.8  &  29.4\\
& ric bd & \multicolumn{5}{l|}{}& & 39.1\\
\hline

\multirow{2}{*}{10} & $\omega$ bd & \multicolumn{6}{l|}{} & 42.8\\
& ric bd & \multicolumn{6}{l|}{}& 35.6\\
\hline

\multirow{2}{*}{11} & $\omega$ bd & \multicolumn{6}{l|}{} & 72.7\\
& ric bd & \multicolumn{6}{l|}{}& 134.1\\
\hline

\multirow{2}{*}{12} & $\omega$ bd & \multicolumn{6}{l|}{} & 195.1\\
& ric bd & \multicolumn{6}{l|}{}& \\
\hline
\end{tabular}\label{tbl:HadamardBPRicOmega}
\end{center}
\end{table}

\begin{table}
\caption{Comparison of the $\omega_2$ based bounds and the RIC based bounds on the $\ell_2$ norms of the errors of the Basis Pursuit algorithm for a Gaussian matrix with leading dimension $n = 256$.}
\begin{center}
\hskip -.3 cm

\begin{tabular}{||l|l||l|l|l|l|l|l|l|}
\cline{2-9}
\multicolumn{1}{l|}{\multirow{2}{*}{}} & $m$ & 51 & 77 & 102 & 128 & 154 & 179 & 205\\
\cline{2-9}
\multicolumn{1}{l|}{}& $s_*$ & 4.6 & 6.2 & 8.1& 9.9 & 12.5& 15.6& 20.0\\
\hline
$k$ & $k_*$ & 2 & 3 & 4 & 4 & 6 & 7 & 10\\
\hline\hline

\multirow{2}{*}{1} & $\omega$ bd & 4.3 & 3.7 & 3.5 & 3.4 & 3.3 & 3.2 & 3.2\\
& ric bd & & & 26.0 & 14.2 & 10.0 & 10.9 & 12.1\\
\hline

\multirow{2}{*}{2} & $\omega$ bd & 34.3 & 12.3 & 8.3 & 7.0 & 6.4 & 5.9 & 5.6\\
& ric bd & & & & & & 47.1 & 27.6\\
\hline

\multirow{2}{*}{3} & $\omega$ bd &  & 197.4 & 23.4 & 14.5 & 11.6 & 9.8 & 8.9\\
& ric bd & & & & & & & \\
\hline

\multirow{2}{*}{4} & $\omega$ bd & \multicolumn{2}{l|}{} &1036.6& 39.6 &   21.7 &   15.9  &  13.4\\
& ric bd & \multicolumn{2}{l|}{}& & & & & \\
\hline

\multirow{2}{*}{5} & $\omega$ bd & \multicolumn{4}{l|}{} & 49.3  &  26.4 &   20.0\\
& ric bd & \multicolumn{4}{l|}{} & & & \\
\hline

\multirow{2}{*}{6} & $\omega$ bd & \multicolumn{4}{l|}{} & 284.2 &   48.8 &   31.2\\
& ric bd & \multicolumn{4}{l|}{} & & & \\
\hline

\multirow{2}{*}{7} & $\omega$ bd & \multicolumn{5}{l|}{} & 129.1  &  48.1\\
& ric bd & \multicolumn{5}{l|}{} & & \\
\hline

\multirow{2}{*}{8} & $\omega$ bd & \multicolumn{6}{l|}{} & 185.5\\
& ric bd & \multicolumn{6}{l|}{} & \\
\hline

\multirow{2}{*}{9} & $\omega$ bd & \multicolumn{6}{l|}{} & 9640.3\\
& ric bd & \multicolumn{6}{l|}{} & \\
\hline
\end{tabular}\label{tbl:GaussianBPRicOmega}
\end{center}
\end{table}

We next compare the error bounds \eqref{eqn:ds_omega_bd} and \eqref{eqn:ds_rip_bd} for the Dantzig selector. For the Bernoulli and Gaussian matrices, our bounds work for wider ranges of $(k,m)$ pairs and are tighter in all tested cases. For the Hadamard matrices, the RIC bounds are better, starting from $k \geq 5$ or $6$. We expect that this indicates a general trend, namely, when $k$ is relatively small, the $\omega$ based bounds are better, while when $k$ is large, the RIC bounds are tighter. This was suggested by the probabilistic analysis of $\omega$ in Section \ref{sec:random}. The reason is that when $k$ is relatively small, both the relaxation $\|\x\|_1 \leq 2k \|\x\|_\infty$ on the sufficient and necessary condition \eqref{nullspaceproperty} and the relaxation $\|\hx - \x\|_2 \leq \sqrt{2k} \|\hx - \x\|_\infty$ are sufficiently tight.

\begin{table}
\caption{Comparison of the $\omega_\infty$ based bounds and the RIC based bounds on the $\ell_2$ norms of the errors of the Dantzig selector algorithm for the Bernoulli matrix used in Table \ref{tbl:BernoulliBPRicOmega}.}
\begin{center}
\hskip 0 cm
\begin{tabular}{||l|l||l|l|l|l|l|l|l|}
\cline{2-9}
\multicolumn{1}{l|}{\multirow{2}{*}{}} & $m$ & 51 & 77 & 102 & 128 & 154 & 179 & 205\\
\cline{2-9}
\multicolumn{1}{l|}{}& $s_*$ & 4.6 & 6.1 & 7.4& 9.6 & 12.1& 15.2& 19.3\\
\hline
$k$ & $k_*$ & 2 & 3 & 3 & 4 & 6 & 7 & 9\\
\hline\hline

\multirow{2}{*}{1} & $\omega$ bd & 6.0 & 5.4 & 4.8 & 4.4 & 4.2 & 4.1 & 4.1\\
& ric bd & & 46.3 & 17.4 & 12.1 & 11.2 & 10.3 & 8.6\\
\hline

\multirow{2}{*}{2} & $\omega$ bd & 102.8 & 38.4 & 29.0 & 18.5 & 14.1 & 12.8 & 11.9\\
& ric bd & & & & & & 47.2 & 22.5\\
\hline

\multirow{2}{*}{3} & $\omega$ bd &  & 1477.2 & 170.2 & 81.2 & 57.0 & 41.1 & 32.6\\
& ric bd & & & & & & & \\
\hline

\multirow{2}{*}{4} & $\omega$ bd & \multicolumn{3}{l|}{} & 522.7 &   194.6 &   128.9  &  89.0\\
& ric bd & \multicolumn{3}{l|}{} & & & & \\
\hline

\multirow{2}{*}{5} & $\omega$ bd & \multicolumn{4}{l|}{} & 768.7  &  323.6 &   203.2\\
& ric bd & \multicolumn{4}{l|}{} & & & \\
\hline

\multirow{2}{*}{6} & $\omega$ bd & \multicolumn{4}{l|}{} & 24974.0 &   888.7 &   489.0\\
& ric bd & \multicolumn{4}{l|}{} & & & \\
\hline

\multirow{2}{*}{7} & $\omega$ bd & \multicolumn{5}{l|}{} & 3417.3  &  1006.9\\
& ric bd & \multicolumn{5}{l|}{} & & \\
\hline

\multirow{2}{*}{8} & $\omega$ bd & \multicolumn{6}{l|}{} & 2740.0\\
& ric bd & \multicolumn{6}{l|}{} & \\
\hline

\multirow{2}{*}{9} & $\omega$ bd & \multicolumn{6}{l|}{} & 10196.9\\
& ric bd & \multicolumn{6}{l|}{} & \\
\hline
\end{tabular}\label{tbl:BernoulliDantzigRicOmega}
\end{center}
\end{table}

\begin{table}
\caption{Comparison of the $\omega_\infty$ based bounds and the RIC based bounds on the $\ell_2$ norms of the errors of the Dantzig selector algorithm for the Hadamard matrix used in Table \ref{tbl:HadamardBPRicOmega}.}
\begin{center}
\hskip -1.2 cm
\begin{tabular}{||l|l||l|l|l|l|l|l|l|}
\cline{2-9}
\multicolumn{1}{l|}{\multirow{2}{*}{}} & $m$ & 51 & 77 & 102 & 128 & 154 & 179 & 205\\
\cline{2-9}
\multicolumn{1}{l|}{}& $s_*$ & 5.2 & 6.9 & 9.1 & 12.1 & 14.4 & 18.3 & 25.2 \\
\hline
$k$ & $k_*$ & 2 & 3 & 4 & 6 & 7 & 9 & 12\\
\hline\hline

\multirow{2}{*}{1} & $\omega$ bd & 4.8 & 4.0 & 3.8 & 3.4 & 3.4 & 3.2 & 3.1\\
& ric bd & &   15.6  &   9.3   &  7.0 &    6.3   &  5.8 &    5.1\\
\hline

\multirow{2}{*}{2} & $\omega$ bd & 50.9  &   16.2  &   10.1 &    7.1  &   7.0  &   6.1 &   5.3\\
& ric bd & & & 45.3  &  16.6  &  13.7  &  10.6   &  8.8\\
\hline

\multirow{2}{*}{3} & $\omega$ bd &  & 108.2 &   30.7 &   14.3  &   13.9  &  10.0 &    8.0\\
& ric bd & & & 1016.4 & 29.9   & 24.9 &   15.8  &   12.5\\
\hline

\multirow{2}{*}{4} & $\omega$ bd & \multicolumn{2}{l|}{} & 150.7 &   35.3 &   29.3  &  16.8   &  11.7\\
& ric bd & \multicolumn{2}{l|}{} & & 126.4 & 38.7 &   24.2 &   16.6\\
\hline

\multirow{2}{*}{5} & $\omega$ bd & \multicolumn{3}{l|}{} & 108.5   & 64.2 &   31.4 &   17.3\\
& ric bd & \multicolumn{3}{l|}{} & & 187.3 & 30.0  &  22.1\\
\hline

\multirow{2}{*}{6} & $\omega$ bd & \multicolumn{3}{l|}{} &3168.9 & 171.5 &  59.7  &  25.3\\
& ric bd & \multicolumn{3}{l|}{} & & 112.0 & 53.1  & 26.8\\
\hline

\multirow{2}{*}{7} & $\omega$ bd & \multicolumn{4}{l|}{} & 1499.5 & 116.3  &  38.8\\
& ric bd & \multicolumn{4}{l|}{} & 411.7 & 71.3 & 34.7\\
\hline

\multirow{2}{*}{8} & $\omega$ bd & \multicolumn{5}{l|}{} & 265.3  &  61.4\\
& ric bd & \multicolumn{5}{l|}{}& 95.4 &   47.6\\
\hline

\multirow{2}{*}{9} & $\omega$ bd & \multicolumn{5}{l|}{} & 2394.0 &   96.0\\
& ric bd & \multicolumn{5}{l|}{}& 198.7  &  61.9\\
\hline

\multirow{2}{*}{10} & $\omega$ bd & \multicolumn{6}{l|}{} & 157.4\\
& ric bd & \multicolumn{6}{l|}{}& 82.9\\
\hline

\multirow{2}{*}{11} & $\omega$ bd & \multicolumn{6}{l|}{} & 296.4\\
& ric bd & \multicolumn{6}{l|}{}& 130.3\\
\hline

\multirow{2}{*}{12} & $\omega$ bd & \multicolumn{6}{l|}{} & 898.2\\
& ric bd & \multicolumn{6}{l|}{}& 201.2 \\
\hline
\end{tabular}\label{tbl:HadamardDantzigRicOmega}
\end{center}
\end{table}

\begin{table}
\caption{Comparison of the $\omega_\infty$ based bounds and the RIC based bounds on the $\ell_2$ norms of the errors of the Dantzig selector algorithm for the Gaussian matrix used in Table \ref{tbl:GaussianBPRicOmega}.}
\begin{center}
\hskip -1.2 cm
\begin{tabular}{||l|l||l|l|l|l|l|l|l|}
\cline{2-9}
\multicolumn{1}{l|}{\multirow{2}{*}{}} & $m$ & 51 & 77 & 102 & 128 & 154 & 179 & 205\\
\cline{2-9}
\multicolumn{1}{l|}{}& $s_*$ & 4.6 & 6.2 & 8.1& 9.9 & 12.5& 15.6& 20.0\\
\hline
$k$ & $k_*$ & 2 & 3 & 4 & 4 & 6 & 7 & 10\\
\hline\hline

\multirow{2}{*}{1} & $\omega$ bd & 6.5 & 5.1 & 4.8 & 4.3 & 4.2 & 4.0 & 3.9\\
& ric bd & &30.0 & 18.0 & 14.6 & 9.7 & 9.3 & 9.1\\
\hline

\multirow{2}{*}{2} & $\omega$ bd & 119.4 & 37.8 & 22.5 & 17.6 & 14.1 & 12.7 & 11.4\\
& ric bd & & & & & 91.5 & 44.4 & 23.5\\
\hline

\multirow{2}{*}{3} & $\omega$ bd &  & 1216.7 & 120.7 & 67.3 & 53.6 & 38.7 & 36.4\\
& ric bd & & & & & & & 2546.6\\
\hline

\multirow{2}{*}{4} & $\omega$ bd & \multicolumn{2}{l|}{} &4515.9& 318.2 &  168.4 &   115.8  &  109.0\\
& ric bd & \multicolumn{2}{l|}{}& & & & & \\
\hline

\multirow{2}{*}{5} & $\omega$ bd & \multicolumn{4}{l|}{} & 663.6  &  292.4 &   247.8\\
& ric bd & \multicolumn{4}{l|}{} & & & \\
\hline

\multirow{2}{*}{6} & $\omega$ bd & \multicolumn{4}{l|}{} & 5231.4 &   764.3 &   453.5\\
& ric bd & \multicolumn{4}{l|}{} & & & \\
\hline

\multirow{2}{*}{7} & $\omega$ bd & \multicolumn{5}{l|}{} & 2646.4  &  1087.7\\
& ric bd & \multicolumn{5}{l|}{} & & \\
\hline

\multirow{2}{*}{8} & $\omega$ bd & \multicolumn{6}{l|}{} & 2450.5\\
& ric bd & \multicolumn{6}{l|}{} & \\
\hline

\multirow{2}{*}{9} & $\omega$ bd & \multicolumn{6}{l|}{} & 6759.0\\
& ric bd & \multicolumn{6}{l|}{} & \\
\hline
\end{tabular}\label{tbl:GaussianDantzigRicOmega}
\end{center}
\end{table}

In Table \ref{tbl:timeomega} we present the execution times for computing different $\omega$. For random matrices with leading dimension $n = 256$, the algorithm generally takes 1 to 3 minutes to compute either $\omega_2(A,s)$ or $\omega_\infty(A^TA,s)$.
\begin{table}
\caption{Time in seconds taken to compute $\omega_2(A,\cdot)$ and $\omega_\infty(A^TA,\cdot)$ for Bernoulli, Hadamard, and Gaussian matrices}
\begin{center}
\hskip 0 cm
\begin{tabular}{||l|l|l|l|l|l|l|l|l|l|}
\hline
$k$ & type & $m$ & 51 & 77 & 102 & 128 & 154 & 179 & 205\\
\hline

\multirow{6}{*}{1} & \multirow{2}{*}{Bernoulli} & $\omega_2$ & 118  &    84  &   133  &    87  &   133   &  174  &   128\\
 & & $\omega_\infty$ & 75  &    81   &   84   &   65  &    63  &   144  &   151\\
\cline{2-10}

 & \multirow{2}{*}{Hadamard} & $\omega_2$ & 84  &    82   &   82   &   82  &    80 &     79  &    79\\
 & & $\omega_\infty$ & 57   &   55  &    58  &    58 &     58   &   58  &    57\\
\cline{2-10}
 & \multirow{2}{*}{Gaussian} & $\omega_2$ & 82   &   84   &  212  &   106 &    156  &   185  &   104\\
 & & $\omega_\infty$ & 69  &    65  &    72  &   102   &   81  &   104  &    72\\
\hline

%
%

\multirow{6}{*}{3} & \multirow{2}{*}{Bernoulli} & $\omega_2$ &   &   155  &    96   &   95  &    97  &    97  &   131\\
 & & $\omega_\infty$ &   &    300  &    228 &     190  &    125  &    135  &    196\\
\cline{2-10}

 & \multirow{2}{*}{Hadamard} & $\omega_2$ &   &   91  &    88  &    87  &    88  &    74   &   72\\
 & & $\omega_\infty$ &    &   84   &   83  &    77  &    92  &   102   &   70\\
\cline{2-10}

 & \multirow{2}{*}{Gaussian} & $\omega_2$ &   &  134  &   168  &   115  &    95  &    96   &  100\\
 & & $\omega_\infty$ &    &  137  &   142  &   125   &  165   &  145   &  105\\
\hline

\multirow{6}{*}{5} & \multirow{2}{*}{Bernoulli} & $\omega_2$ & \multicolumn{3}{l}{} & &  97  &   111  &    97\\
 & & $\omega_\infty$ &  \multicolumn{3}{l}{} & & 156   &   81  &   107\\
\cline{2-10}

 & \multirow{2}{*}{Hadamard} & $\omega_2$ & \multicolumn{3}{l|}{} & 87   &   85   &   85   &   81\\
 & & $\omega_\infty$ & \multicolumn{3}{l|}{}  & 75   &   74   &   75   &   75\\
\cline{2-10}

 & \multirow{2}{*}{Gaussian} & $\omega_2$ & \multicolumn{3}{l}{} & &  98  &   105   &   96\\
 & & $\omega_\infty$ & \multicolumn{3}{l}{} & & & &     193\\
\hline

\multirow{6}{*}{7} & \multirow{2}{*}{Bernoulli} & $\omega_2$ & \multicolumn{4}{l|}{} & & 164 & 104\\
 & & $\omega_\infty$ & \multicolumn{4}{l|}{} & & 178 & 85\\
\cline{2-10}

 & \multirow{2}{*}{Hadamard} & $\omega_2$ & \multicolumn{4}{l|}{} & & 82 & 77\\\
 & & $\omega_\infty$ & \multicolumn{4}{l|}{} & 134 & 71 & 65\\
\cline{2-10}

 & \multirow{2}{*}{Gaussian} & $\omega_2$ & \multicolumn{4}{l|}{} & & 106 & 105\\
 & & $\omega_\infty$ & \multicolumn{4}{l|}{} & & & 193\\
\hline

\end{tabular}\label{tbl:timeomega}
\end{center}
\end{table}


In the last set of experiments, we compute $\omega_2(A,2k)$ and $\omega_\infty(A^TA,2k)$ for a Gaussian matrix and a Hadamard matrix, respectively, with leading dimension $n = 512$. The row dimensions of the sensing matrices range over $m = \lfloor \rho n \rfloor$ with $\rho = 0.2, 0.3, \ldots, 0.8$.
In Figure \ref{fig:GaussianRICvsOmeBD}, we compare the $\ell_2$ norm error bounds of the Basis Pursuit using $\omega_2(A,2k)$ and the RIC. The color indicates the values of the error bounds. We remove all bounds that are greater than 50 or are not valid. Hence, all white areas indicate that the bounds corresponding to $(k,m)$ pairs that are too large or not valid. The left sub-figure is based on $\omega_2(A,2k)$ and the right sub-figure is based on the RIC. We observe that the $\omega_2(A,2k)$ based bounds apply to a wider range of $(k,m)$ pairs.
\begin{figure}
\hskip -0cm
\centering
\includegraphics[width = 0.7\textwidth, trim = 0mm 0mm 0mm 0mm, clip]{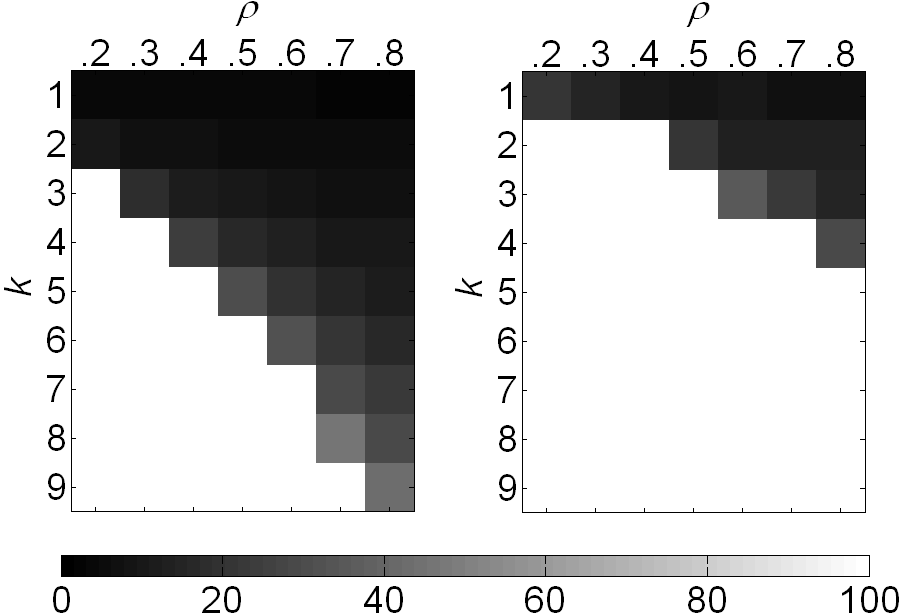}
\caption{$\omega_2(A,2k)$ based bounds v.s. RIC based bounds on the $\ell_2$ norms of the errors for a Gaussian matrix with leading dimension $n = 512$. Left: $\omega_2(A,2k)$ based bounds; Right: RIC based bounds.}
\label{fig:GaussianRICvsOmeBD}
\end{figure}%

In Figure \ref{fig:HadamardRICvsOmeBD}, we conduct the same experiment as in Figure \ref{fig:GaussianRICvsOmeBD} for a Hadamard matrix and the Dantzig selector. We observe that for the Hadamard matrix, the RIC gives better performance bounds. This result coincides with the one we obtained in Table \ref{tbl:HadamardDantzigRicOmega}.

The average time for computing each $\omega_2(A,2k)$ and $\omega_\infty(A^TA,2k)$ was around 15 minutes.
\begin{figure}
\hskip -0cm
\centering
\includegraphics[width = 0.7\textwidth, trim = 0mm 0mm 0mm 0mm, clip]{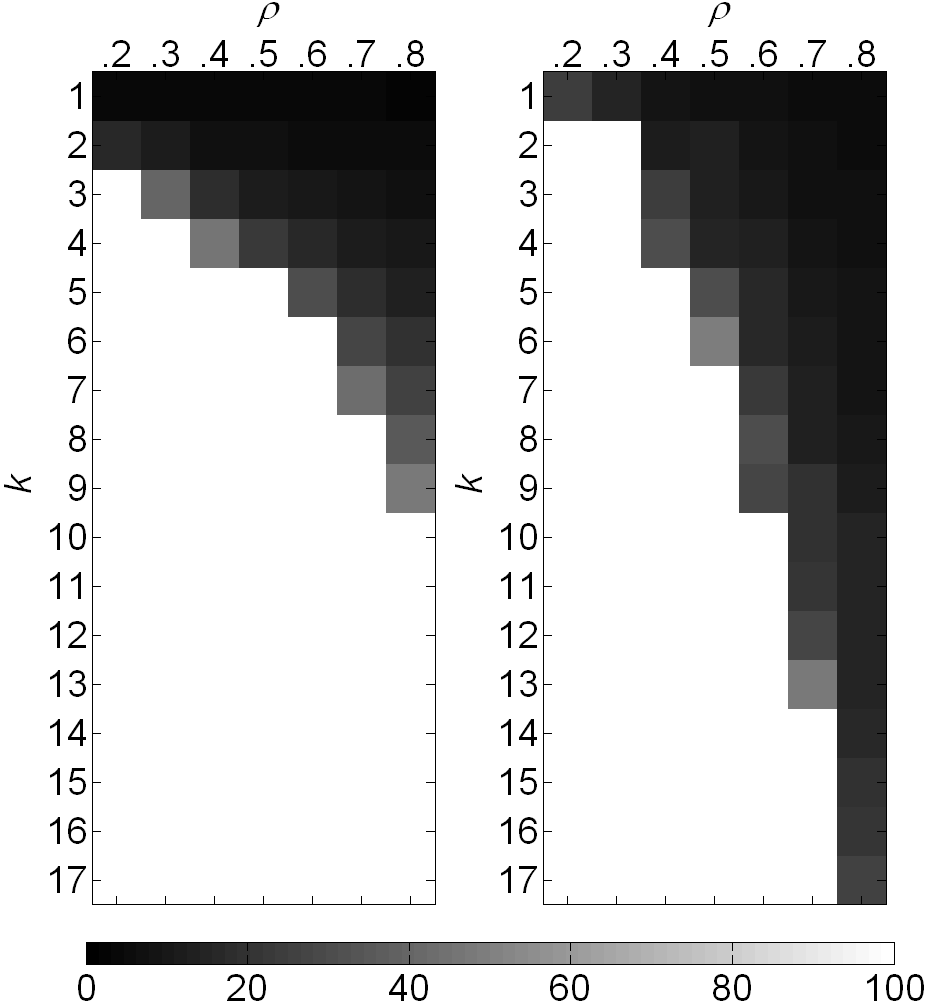}
\caption{$\omega_\infty(A^TA,2k)$ based bounds v.s. RIC based bounds on the $\ell_2$ norms of the errors for a Hadamard matrix with leading dimension $n = 512$. Left: $\omega_2(A,2k)$ based bounds; Right: RIC based bounds}
\label{fig:HadamardRICvsOmeBD}
\end{figure}%
\section{Conclusions}\label{sec:conclusions}
In this paper, we analyzed the performance of $\ell_1$ sparse signal recovery algorithms using the $\ell_\infty$ norm of the errors as a performance criterion. We expressed other popular performance criteria in terms of the $\ell_\infty$ norm. A family of goodness measures of the sensing matrices was defined using optimization procedures. We used these goodness measures to derive upper bounds on the $\ell_\infty$ norms of the reconstruction errors for the Basis Pursuit, the Dantzig selector, and the LASSO estimator. Polynomial-time algorithms with established convergence properties were implemented to efficiently solve the optimization procedures defining the goodness measures. We expect that these goodness measures will be useful in comparing different sensing systems and recovery algorithms, as well as in designing optimal sensing matrices. In future work, we will use these computable performance bounds to optimally design $k-$space sample trajectories for MRI and to optimally design transmitting waveforms for compressive sensing radar.

\section{Appendix: Proofs}
\subsection{Proof of Proposition \ref{pro:errorcharacteristics}}\label{app:pf:errorcharacteristics}

\begin{proof}[Proof of Proposition \ref{pro:errorcharacteristics}] Suppose $S = \supp(\x)$ and $|S| = \|\x\|_0 = k$. Define the error vector $\h = \hx - \x$. For any vector $\z \in \R^n$ and any index set $S \subseteq \{1,\ldots,n\}$, we use $\z_S \in \R^{|S|}$ to represent the vector whose elements are those of $\z$ indicated by $S$.\\

We first deal with the Basis Pursuit and the Dantzig selector. As observed by Cand\'es in \cite{Candes2008RIP},  the fact that $\|\hx\|_1 = \|\x + \h\|_1$ is the minimum among all $\z$s satisfying the constraints in \eqref{bp} and \eqref{ds}, together with the fact that the true signal $\x$ satisfies the constraints as required by the conditions imposed on the noise in Proposition \ref{pro:errorcharacteristics}, imply that $\|\h_{S^c}\|_1$ cannot be very large. To see this, note that
\begin{eqnarray}\label{x_min}
  \|\x\|_1 &\geq& \|\x + \h\|_1\nonumber\\
  & = & \sum_{i\in S}|\x_i + \h_i| + \sum_{i\in S^c}|\x_i + \h_i| \nonumber\\
  &\geq& \|\x_S\|_1 - \|\h_S\|_1 + \|\h_{S^c}\|_1\nonumber\\
  & = & \|\x\|_1 - \|\h_S\|_1 + \|\h_{S^c}\|_1.
\end{eqnarray}
Therefore, we obtain $\|\h_{S}\|_1 \geq \|\h_{S^c}\|_1$, which leads to
\begin{eqnarray}\label{h1h2}
2\|\h_S\|_1  \geq \|\h_S\|_1 + \|\h_{S^c}\|_1 =  \|\h\|_1.
\end{eqnarray}
\\
We now turn to the LASSO estimator\eqref{lasso}. We use the proof technique in \cite{candes2009lowrank} (see also \cite{bickel2009simultaneous}). Since the noise $\w$ satisfies $\|A^T\w\|_\infty \leq \kappa\mu$ for some small $\kappa > 0$, and $\hx$ is a solution to \eqref{lasso}, we have
\begin{eqnarray*}
\frac{1}{2}\|A\hx - \y\|_2^2 + \mu \|\hx\|_1 \leq \frac{1}{2} \|A\x - \y\|_2^2 + \mu\|\x\|_1.
\end{eqnarray*}
Consequently, substituting $\y = A\x + \w$ yields
\begin{eqnarray*}
\mu \|\hx\|_1 &\leq& \frac{1}{2} \|A\x - \y\|_2^2 - \frac{1}{2}\|A\hx - \y\|_2^2 + \mu \|\x\|_1 \nn\\
&=& \frac{1}{2}\|\w\|_2^2 - \frac{1}{2}\|A(\hx-\x) - \w\|_2^2 + \mu\|\x\|_1\nn\\
& = & \frac{1}{2}\|\w\|_2^2 - \frac{1}{2}\|A(\hx-\x)\|_2^2\nn\\
  && \ \ \ + \left<A(\hx-\x), \w\right> - \frac{1}{2}\|\w\|_2^2 + \mu\|\x\|_1\nn\\
&\leq & \left<A(\hx-\x), \w\right> + \mu\|\x\|_1\nn\\
& = & \left<\hx - \x, A^T\w\right> + \mu\|\x\|_1.
\end{eqnarray*}
Using the Cauchy-Swcharz type inequality, we get
\begin{eqnarray*}
\mu \|\hx\|_1 &\leq & \|\hx-\x\|_1\|A^T\w\|_\infty + \mu\|\x\|_1\nn\\
& = & \kappa \mu \|\h\|_1 + \mu\|\x\|_1,
\end{eqnarray*}
which leads to
\begin{eqnarray*}
  \|\hx\|_1 &\leq & \kappa\|\h\|_1 + \|\x\|_1.
\end{eqnarray*}
Therefore, similar to the argument in \eqref{x_min}, we have
\begin{eqnarray*}
 &&\|\x\|_1 \nn\\
 &\geq& \|\hx\|_1 - \kappa\|\h\|_1\nn\\
 & = & \|\x + \h_{S^c} + \h_S\|_1- \kappa\left(\|\h_{S^c} + \h_S\|_1\right) \nonumber\\
  &\geq& \|\x + \h_{S^c} \|_1 - \|\h_S\|_1 - \kappa\left(\|\h_{S^c}\|_1 + \|\h_S\|_1\right) \nonumber\\
  & = & \|\x\|_1 + (1-\kappa)\|\h_{S^c}\|_1 - (1+\kappa)\|\h_S\|_1,
\end{eqnarray*}
where $S = \supp(\x)$.
Consequently, we have
\begin{eqnarray*}
  \|\h_{S}\|_1 &\geq& \frac{1-\kappa}{1+\kappa}\|\h_{S^c}\|_1.
\end{eqnarray*}
Therefore, similar to \eqref{h1h2}, we obtain
\begin{eqnarray}\label{lassoh1h2}
\frac{2}{1-\kappa}\|\h_S\|_1 &\geq& \frac{1+\kappa}{1-\kappa}\|\h_S\|_1 + \frac{1-\kappa}{1-\kappa} \|\h_{S}\|_1\nn\\
&\geq& \frac{1+\kappa}{1-\kappa}\frac{1-\kappa}{1+\kappa}\|\h_{S^c}\|_1 + \frac{1-\kappa}{1-\kappa} \|\h_{S}\|_1\nn\\
& = & \|\h\|_1.
\end{eqnarray}
\end{proof}

\subsection{Proof of Theorem \ref{thm:fix_fs}}\label{app:pf:fix_fs}
\begin{proof}
\begin{enumerate}
\item Since in the optimization problem defining $f_{s,i}(\eta)$, the objective function $\z_i$ is continuous, and the constraint correspondence
\begin{eqnarray}
  C(\eta): &&[0, \infty) \twoheadrightarrow \R^{n}\nonumber\\
  &&\eta \mapsto \left\{\z: \|Q\z\|_\diamond \leq 1, {\|\z\|_{1}} \leq s \eta\right\}
\end{eqnarray}
is compact-valued and continuous (both upper and lower hemicontinuous), according to Berge's Maximum Theorem \cite{Berge1997maximum}, the optimal value function $f_{s,i}(\eta)$ is continuous. The continuity of $f_s(\eta)$ follows from that finite maximization preserves the continuity.

\item  To show the strict increasing property, suppose $0 < \eta_1 < \eta_2$ and the dual variable ${\bs \lambda}_2^*$ achieves $f_{s,i}(\eta_2)$ in \eqref{eqn:fsi_dual}. Then we have
\begin{eqnarray}
  f_{s,i}(\eta_1) &\leq&  s\eta _1 \|\e_i - Q^T{\bs \lambda}_2^*\|_\infty + \|{\bs \lambda}_2\|_\diamond^*\nonumber\\
  &<& s\eta _2 \|\e_i - Q^T{\bs \lambda}_2^*\|_\infty + \|{\bs \lambda}_2\|_\diamond^*\nonumber\\
  &=& f_{s,i}(\eta_2).
\end{eqnarray}
The case for $\eta_1 = 0$ is proved by continuity, and the strict increasing of $f_s(\eta)$ follows immediately.

\item The concavity of $f_{s,i}(\eta)$ follows from the dual representation \eqref{eqn:fsi_dual} and the fact that $f_{s,i}(\eta)$ is the minimization of a function of variables $\eta$ and $\bs \lambda$, and when $\bs\lambda$, the variable to be minimized, is fixed, the function is linear in $\eta$.

\item Next we show that when $\eta > 0$ is sufficiently small $f_s(\eta) \geq s \eta$. Taking $\z = s\eta \e_i$, we have $\|\z\|_{1} = s\eta$ and $\z_i = s\eta > \eta$ (recall $s \in (1, \infty)$). In addition, when $0 < \eta \leq 1/(s\|Q_i\|_\diamond)$, we also have $\|Q\z\|_\diamond \leq 1$. Therefore, for sufficiently small $\eta$, we have $f_{s, i}(\eta) \geq s\eta > \eta$. Clearly, $f_s(\eta) = \max_i f_{s,i}(\eta) \geq s\eta > \eta$ for such $\eta$.

Recall that
\begin{eqnarray}
  \frac{1}{s_*} &=&  \max_i\min_{{\bs \lambda}_i} \|\e_i - Q^T{\bs \lambda}_i\|_\infty.
\end{eqnarray}
Suppose ${\bs \lambda}_i^*$ is the optimal solution for each $\min_{{\bs \lambda}_i}\|\e_i - Q^T{\bs \lambda}_i\|_\infty$. For each $i$, we then have
\begin{eqnarray}
  \frac{1}{s_*} &\geq & \|\e_i - Q^T{\bs \lambda}_i^*\|_\infty,
\end{eqnarray}
which implies
\begin{eqnarray}
  f_{s,i}(\eta) &=& \min_{{\bs \lambda}_i} s\eta  \|\e_i - Q^T{\bs \lambda}_i\|_\infty + \|{\bs \lambda}_i\|_\diamond^* \nonumber\\
  &\leq & s\eta  \|\e_i - Q^T{\bs \lambda}_i^*\|_\infty + \|{\bs \lambda}_i^*\|_\diamond^* \nonumber\\
  &\leq & \frac{s}{s_*} \eta + \|{\bs \lambda}_i^*\|_\diamond^*.
\end{eqnarray}

As a consequence, we obtain
\begin{eqnarray}
f_s(\eta) = \max_i f_{s,i}(\eta) \leq \frac{s}{s_*} \eta + \max_i \|{\bs \lambda}_i^*\|_\diamond^*.
\end{eqnarray}
Pick $\rho \in (s/s_*, 1)$. Then, we have the following when $\eta > \max_i \|{\bs \lambda}_i^*\|_\diamond^*/(\rho - s/s_*)$:
\begin{eqnarray}
  f_{s,i}(\eta) &\leq& \rho \eta, i = 1, \ldots, n, \text{\ and \ } \nn\\
  f_s(\eta) &\leq& \rho \eta.
\end{eqnarray}

\item We first show the existence and uniqueness of the positive fixed points for $f_{s,i}(\eta)$. The properties 1) and 4) imply that $f_{s,i}(\eta)$ has at least one positive fixed point. (Interestingly, 2) and 4) also imply the existence of a positive fixed point, see \cite{tarski1955fixedpoint}.) To prove uniqueness, suppose there are two fixed points $0 < \eta_1^* < \eta_2^*$. Pick $\eta_0$ small enough such that $f_{s,i}(\eta_0) > \eta_0 > 0$ and $\eta_0 < \eta_1^*$. Then $\eta_1^* = \lambda \eta_0 + (1-\lambda)\eta_2^*$ for some $\lambda \in (0, 1)$, which implies that $f_{s,i}(\eta_1^*) \geq \lambda f_{s,i}(\eta_0) + (1-\lambda) f_{s,i}(\eta_2^*) > \lambda \eta_0 + (1-\lambda) \eta_2^* = \eta_1^*$ due to the concavity, contradicting with $\eta_1^* = f_{s,i}(\eta_1^*)$.

    The set of positive fixed point for $f_s(\eta)$, $\{\eta \in (0, \infty): \eta = f_s(\eta) = \max_i f_{s,i}(\eta)\}$, is a subset of $\bigcup_{i=1}^p \{\eta \in (0, \infty): \eta = f_{s,i}(\eta) \} = \{\eta_i^*\}_{i=1}^n$. We argue that
    \begin{eqnarray}
     \eta^* = \max_i \eta_i^*
    \end{eqnarray}
is the unique positive fixed point for $f_s(\eta)$.

We proceed to show that $\eta^*$ is a fixed point of $f_s(\eta)$. Suppose $\eta^*$ is a fixed point of $f_{s,i_0}(\eta)$, then it suffices to show that $f_s(\eta^*) = \max_i f_{s,i}(\eta^*) = f_{s,i_0}(\eta^*)$. If this is not the case, there exists $i_1 \neq i_0$ such that $f_{s,i_1}(\eta^*) > f_{s,i_0}(\eta^*) = \eta^*$. The continuity of $f_{s,i_1}(\eta)$ and the property 4) imply that there exists $\eta > \eta^*$ with $f_{s,i_1}(\eta) = \eta$, contradicting with the definition of $\eta^*$.

To show the uniqueness, suppose $\eta_1^*$ is fixed point of $f_{s,i_1}(\eta)$ satisfying $\eta_1^* < \eta^*$. Then, we must have $f_{s,i_0}(\eta_1^*) > f_{s,i_1}(\eta_1^*)$ because otherwise the continuity implies the existence of another fixed point of $f_{s,i_0}(\eta)$. As a consequence, $f_s(\eta_1^*) > f_{s,i_1}(\eta_1^*) = \eta_1^*$ and $\eta_1^*$ is not a fixed point of $f_s(\eta)$.

\item Next we show $\eta^*  = \gamma^* \df 1/\omega_\diamond(Q,s)$. We first prove $\gamma^* \geq \eta^*$ for the fixed point $\eta^* = f_s(\eta^*)$. Suppose $\z^*$ achieves the optimization problem defining $f_s(\eta^*)$, then we have
\begin{eqnarray}
\eta^* = f_s(\eta^*) = \|\z^*\|_{\infty},  \|Q\z^*\|_\diamond \leq 1, \text{and\ } \|\z^*\|_{1} \leq s\eta^*.
\end{eqnarray}
Since $\|\z^*\|_{1}/\|\z^*\|_{\infty} \leq s\eta^*/\eta^* \leq s$, we have
\begin{eqnarray}
  \gamma^* &\geq& \frac{\|\z^*\|_{\infty}}{\|Q\z^*\|_\diamond} \geq \eta^*.
\end{eqnarray}

If $\eta^* < \gamma^*$, we define $\eta_0 = (\eta^* + \gamma^*)/2$ and
\begin{eqnarray}
\hskip -1cm &&\z^{\mathrm{c}} = \mathrm{argmax}_{\z}{\frac{s\|\z\|_{\infty}}{\|\z\|_{1}}} \text{\ s.t. \ } \|Q\z\|_\diamond \leq 1, \|\z\|_{\infty} \geq \eta_0,\label{eqn:defzc}\\
\hskip -1cm  &&\rho = {\frac{s\|\z^{\mathrm{c}}\|_{\infty}}{\|\z^{\mathrm{c}}\|_{1}}}.\label{eqn:defrho}
\end{eqnarray}
Suppose $\z^{**}$ with $\|Q\z^{**}\|_\diamond = 1$ achieves the optimum of the optimization \eqref{eqn:max_inf_Q_diamond} defining $\gamma^* = 1/\omega_\diamond(Q,s)$. Clearly, $\|\z^{**}\|_{\infty} = \gamma^* > \eta_0$, which implies $\z^{**}$ is a feasible point of the optimization problem \eqref{eqn:defzc}  defining $\z^{\mathrm{c}}$ and $\rho$. As a consequence, we have
\begin{eqnarray}
  \rho \geq {\frac{s\|\z^{**}\|_{\infty}}{\|\z^{**}\|_{1}}} \geq 1.
\end{eqnarray}

\begin{figure}[h!t]
\hskip -0cm
\centering
\includegraphics[width = 0.4\textwidth, trim = 0mm 0mm 0mm 0mm, clip]{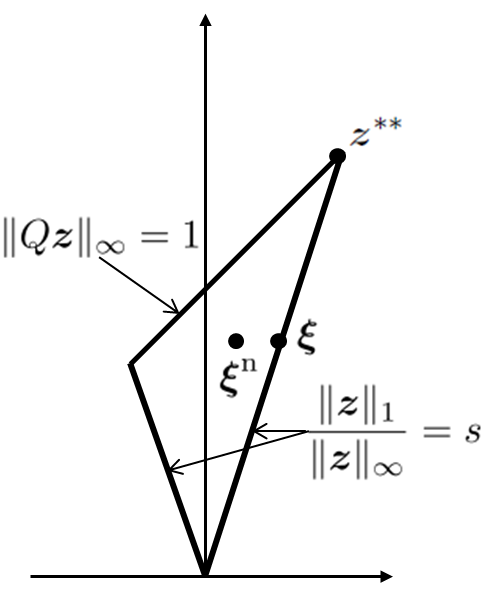}
\caption{Illustration of the proof for $\rho > 1$.}
\label{fig:proof}
\end{figure}%

Actually we will show that $\rho > 1$. If $\|\z^{**}\|_{1} < s\|\z^{**}\|_{\infty}$, we are done. If not (\emph{i.e.}, $\|\z^{**}\|_{1} = s\|\z^{**}\|_{\infty}$), as illustrated in Figure \ref{fig:proof}, we consider $\bs \xi = \frac{\eta_0}{\gamma^*}\z^{**}$, which satisfies
\begin{eqnarray}
  &&\|Q\bs \xi\|_\diamond = \frac{\eta_0}{\gamma^*} < 1,\\
  &&\|\bs \xi\|_{\infty}= \eta_0, \text{\  and \ }\\
  &&\|\bs \xi\|_{1} = s\eta_0.
\end{eqnarray}

To get ${\bs \xi}^{\mathrm{n}}$ as shown in Figure \ref{fig:proof}, pick the component of $\bs \xi$ with the smallest non-zero absolute value, and scale that component by a small positive constant less than $1$. Because $s > 1$, $\bs \xi$ has more than one non-zero components, implying $\|{\bs \xi}^{\mathrm{n}}\|_{\infty}$ will remain the same. If the scaling constant is close enough to $1$, $\|Q{\bs \xi}^{\mathrm{n}}\|_\diamond$ will remain less than 1 due to continuity. But the good news is that $\|{\bs \xi}^{\mathrm{n}}\|_{1}$ decreases, and hence $\rho \geq \frac{s\|{\bs \xi}^{\mathrm{n}}\|_{\infty}}{\|{\bs \xi}^{\mathrm{n}}\|_{1}}$ becomes greater than 1.

Now we proceed to obtain a contradiction that $f_s(\eta^*) > \eta^*$. If $\|\z^{\mathrm{c}}\|_{1} \leq s\cdot \eta^*$, then it is a feasible point of
\begin{eqnarray}\label{eqn:subt0}
  \max_{\z} \|\z\|_{\infty}\text{\ s.t. \ } \|Q\z\|_\diamond \leq 1, \|\z\|_{1} \leq s\cdot \eta^*.
\end{eqnarray}
As a consequence, $f_s(\eta^*) \geq \|\z^{\mathrm{c}}\|_{\infty}\geq \eta_0 > \eta^*$, contradicting with $\eta^*$ is a fixed point and we are done. If this is not the case, \emph{i.e.}, $\|\z^{\mathrm{c}}\|_{1} > s\cdot \eta^*$, we define a new point
\begin{eqnarray}
  \z^{\mathrm{n}} = \tau \z^{\mathrm{c}}
\end{eqnarray}
with
\begin{eqnarray}
  \tau = \frac{s\cdot \eta^*}{\|\z^\mathrm{c}\|_{1}} < 1.
\end{eqnarray}
Note that $\z^{\mathrm{n}}$ is a feasible point of the optimization problem defining $f_s(\eta^*)$ since
\begin{eqnarray}
&&\|Q\z^{\mathrm{n}}\|_\diamond = \tau \|Q\z^{\mathrm{c}}\|_\diamond < 1, \text{\ and \ }\\
&&\|\z^{\mathrm{n}}\|_{ 1} = \tau \|\z^{\mathrm{c}}\|_{1} = s\cdot \eta^*.
\end{eqnarray}
Furthermore, we have
\begin{eqnarray}
  \|\z^{\mathrm{n}}\|_{\infty}= \tau \|\z^{\mathrm{c}}\|_{\infty}= \rho \eta^*.
\end{eqnarray}
As a consequence, we obtain a contradiction
\begin{eqnarray}
  f_s(\eta^*) &\geq& \rho \eta^* > \eta^*.
\end{eqnarray}
\begin{figure}[h!t]
\hskip -0cm
\centering
\includegraphics[width = 0.5\textwidth, trim = 0mm 0mm 0mm 0mm, clip]{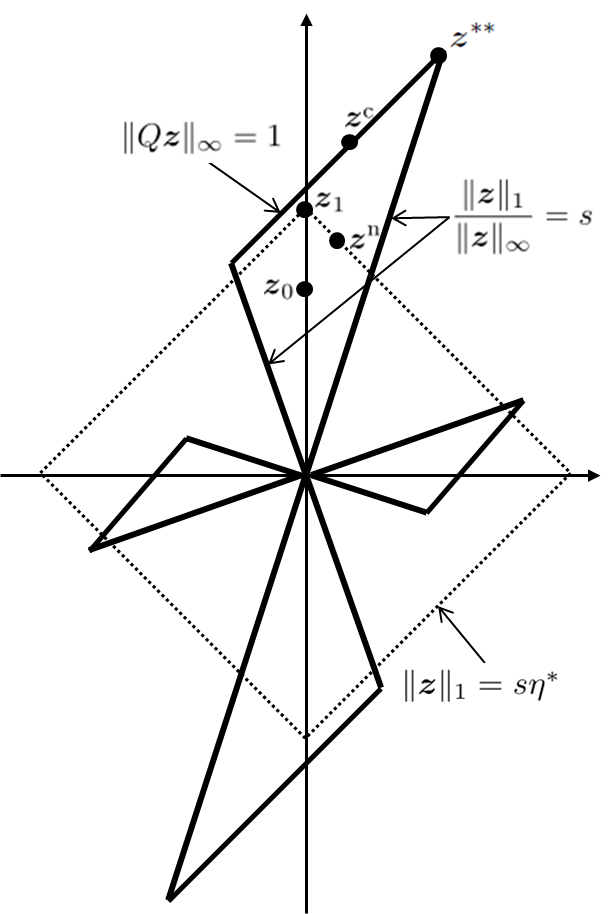}
\caption{Illustration of the proof for $f_s(\eta^*) \geq \rho\eta^*$.}
\label{fig:proof1}
\end{figure}%
Therefore, for the fixed point $\eta^*$, we have $\eta^* = \gamma^* = 1/\omega_\diamond(Q,s)$.

\item This property simply follows from the continuity, the uniqueness, and property 4).
\item We use contradiction to show the existence of $\rho_1(\epsilon)$ in 8). In view of 4), we need only to show the existence of such a $\rho_1(\epsilon)$ that works for $\eta_L \leq \eta \leq (1-\epsilon)\eta^*$ where $\eta_L = \sup\{\eta: f_s(\xi) \geq s\xi, \forall 0 < \xi \leq \eta\}$. Suppose otherwise, we then construct sequences $\{\eta^{(k)}\}_{k=1}^\infty \subset [\eta_L, (1-\epsilon)\eta^*]$ and $\{\rho_1^{(k)}\}_{k=1}^\infty \subset (1, \infty)$ with
\begin{eqnarray}
&& \lim_{k\rightarrow \infty} \rho_1^{(k)} = 1,\nonumber\\
&& f_s(\eta^{(k)}) \leq \rho^{(k)} \eta^{(k)}.
\end{eqnarray}
Due to the compactness of $[\eta_L, (1-\epsilon)\eta^*]$, there must exist a subsequence $\{\eta^{(k_l)}\}_{l=1}^\infty$ of $\{\eta^{(k)}\}$ such that $\lim_{l\rightarrow \infty}\eta^{(k_l)} = \eta_{\mathrm{lim}}$ for some $\eta_{\mathrm{lim}} \in [\eta_L, (1-\epsilon)\eta^*]$. As a consequence of the continuity of $f_s(\eta)$, we have

\begin{eqnarray}
f_s(\eta_{\mathrm{lim}}) =  \lim_{l\rightarrow \infty} f_s(\eta^{(k_l)}) \leq \lim_{l\rightarrow \infty} \rho_1^{(k_l)} \eta^{(k_l)} = \eta_{\mathrm{lim}}.
\end{eqnarray}
Again due to the continuity of $f_s(\eta)$ and the fact that $f_s(\eta) < \eta$ for $\eta < \eta_L$, there exists $\eta_c \in [\eta_L, \eta_{\mathrm{lim}}]$ such that
\begin{eqnarray}
  f_s(\eta_c) &=& \eta_c,
\end{eqnarray}
contradicting with the uniqueness of the fixed point for $f_s(\eta)$. The existence of $\rho_2(\epsilon)$ can be proved in a similar manner.
\end{enumerate}

\end{proof}

\bibliographystyle{siam}
\bibliography{/Users/gongguotang/SugarSync/Papers/Material/BibTex/IEEEabrv,/Users/gongguotang/SugarSync/Papers/Material/BibTex/Gongbib}
\end{document}